\documentclass[letterpaper]{article} 
\usepackage{aaai24}  
\usepackage{times}  
\usepackage{helvet}  
\usepackage{courier}  
\usepackage[hyphens]{url}  
\usepackage{graphicx} 
\usepackage{natbib}  
\usepackage{caption} 
\usepackage[font=small,labelfont=bf]{caption}
\usepackage{enumerate}

\usepackage{algorithm}
\usepackage{algorithmic}

\usepackage{newfloat}
\usepackage{listings}
\DeclareCaptionStyle{ruled}{labelfont=normalfont,labelsep=colon,strut=off} 
\floatstyle{ruled}
\newfloat{listing}{tb}{lst}{}
\floatname{listing}{Listing}
\usepackage{amsmath,amssymb,amsthm}
\usepackage{thmtools,cleveref}
\usepackage{thm-restate}
\usepackage{tikz}
\usetikzlibrary{arrows,positioning,shapes,decorations,automata,backgrounds,petri,fit,calc}
\usepackage{soul}
\usepackage{lipsum}

\usepackage{paralist}

\usepackage{enumitem}

\newlist{inlineenum}{enumerate*}{1}
\setlist*[inlineenum]{mode=unboxed,label=(\arabic*)}

\newtheorem{theorem}{Theorem}[section]

\newtheorem{corollary}[theorem]{Corollary}
\newtheorem{proposition}[theorem]{Proposition}

\newtheorem{problem}{Problem}[section]

\theoremstyle{definition}
\newtheorem{example}{Example}[section]

\newcommand{\commentout}[1]{}

\newcommand{\bivector}[2]{\ensuremath{\left[\begin{smallmatrix} #1 \\ #2 \end{smallmatrix}\right]}}

\newcommand{\query}[1]{\textsc{#1}}
\newcommand{\mq}{\query{mq}}

\newcommand{\dr}{\query{dr}}

\newcommand{\true}{\textsc{t}}
\newcommand{\false}{\textsc{f}}
\newcommand{\sstart}{\textsc{s}}

\newcommand{\algo}[1]{\mathfrak{#1}}

\newcommand{\fsend}{f^{\textsf{!!}}}
\newcommand{\frec}{f^{\textsf{??}}}
\newcommand{\fst}{f^{\textsf{st}}}

\newcommand{\fact}{f^{\textsf{act}}}

\newcommand{\tree}[1]{\mathcal{#1}}

\newcommand{\state}[1]{q_{#1}}
\newcommand{\statep}[1]{p_{#1}}

\newcommand{\scstate}[1]{\textsc{#1}}
\newcommand{\scstatep}[1]{\textsc{#1}\textsc{'}}

\newcommand{\apart}[1]{\,\ensuremath{\#}_{#1}\,}

\title{Learning Broadcast Protocols}
\author{
Dana Fisman\equalcontrib\textsuperscript{\rm 1}\quad
Noa Izsak\equalcontrib\textsuperscript{\rm 1} \quad
Swen Jacobs\equalcontrib\textsuperscript{\rm 2}
}
\affiliations{
\textsuperscript{\rm 1}Ben-Gurion University,\\
\textsuperscript{\rm 2}CISPA Helmholtz Center for Information Security,\\
dana@cs.bgu.ac.il \quad
    izsak@post.bgu.ac.il \quad
    jacobs@cispa.de
}

\usepackage{bibentry}

\begin{document}

\maketitle

\begin{abstract}
The problem of learning a computational model from examples has been receiving growing attention.
For the particularly challenging problem of learning 
models of distributed systems, existing results are restricted to models with a \emph{fixed} number of interacting processes.
In this work we look for the first time (to the best of our knowledge) at the problem of learning a distributed system with an arbitrary number of processes, assuming only that there \emph{exists} a cutoff, i.e., a number of processes that is sufficient to produce all observable behaviors.
Specifically, we consider \emph{fine broadcast protocols}, these are 
broadcast protocols 
(BPs) with a finite cutoff and no hidden states. 
 We provide a learning algorithm that can infer a correct BP from a sample that is consistent with a fine BP, and a minimal equivalent BP if the sample is sufficiently complete.  
 On the negative side we show that (a) characteristic sets of exponential size are unavoidable, (b) 
 the consistency problem for fine BPs is NP hard, and (c)
 that fine BPs are not polynomially predictable.
\end{abstract}

\section{Introduction}
Learning computational models has a long history
starting with the seminal works of Gold~\shortcite{Gold67,Gold78} and Angluin~\shortcite{Angluin87}. Questions regarding learning computational models have raised a lot of interest both in the artificial intelligence community and the verification community. 
(Peled et al.~\shortcite{PeledVY02}, Vaandrager~\shortcite{Vaandrager17}).
Many results regarding the learnability of various computational models used in verification have already been obtained~
\cite{BeimelBBKV00}, \cite{BolligHLM13}, \cite{DeckerHLT14},
Angluin et al.~\shortcite{AngluinEF15}, Balle et al.~\shortcite{BalleM15}, Drews et al.~\shortcite{DrewsD17},
Roy et al.~\shortcite{RFN20}, Vaandrager et al.~\shortcite{VaandragerB021}, Saadon et al.~\shortcite{FismanS22},
 Nitay et al.~\shortcite{FismanNZ23}, Frenkel et al.~\shortcite{FismanFZ23}.

Particularly challenging is learning of concurrent computational models.
Compared to most sequential models, they offer another level of succinctness, and they usually have no unique minimal model.
Both of these aspects can make learning significantly more difficult.
Various results regarding learning concurrent models have already been obtained~\cite{BolligKKL10},\cite{AartsFKV15},
Esparza et al.~\shortcite{EsparzaLS11},
Muscholl et al.~\shortcite{MuschollW22}.
However, these results are limited to models with a fixed number of processes, and therefore cannot reliably learn models for (distributed) protocols that are expected to run correctly for \emph{any} number of processes.

\emph{Broadcast protocols} (BPs) are a powerful concurrent computational model, allowing the synchronous communication of the sender of an action with an arbitrary number of receivers~\cite{DBLP:conf/lics/EmersonN98}. 
In its basic form, this model assumes that communication and processes are reliable, i.e., it does not consider communication failures or faulty processes.
BPs have mainly been studied in the context of parameterized verification, i.e., 
proving functional correctness according to a formal specification, for all systems where an arbitrary number of processes execute a given protocol.

The challenge in reasoning about parameterized systems such as BPs is that a parameterized system concisely represents an infinite family of systems: for each natural number $n$ it includes the system where $n$ indistinguishable
processes interact. 
The system is correct only if it satisfies the specification for any number $n$ of processes interacting.
In the context of verification, a variety of approaches has been investigated to overcome this challenge. 

Some of these approaches are based on the notion of \emph{cutoff}, i.e., a number $c$ of processes such that a given property holds for any instance of the system with $n\!\geq \!c$ processes if and only if it holds for the cutoff system, where the \emph{cutoff system} is a system with exactly $c$ processes interacting.
In the literature, many results exist that provide cutoffs for certain classes of properties in a given computational model~\cite{DBLP:journals/ijfcs/EmersonN03,DBLP:conf/cade/EmersonK00}, Au{\ss}erlechner et al.~\shortcite{DBLP:conf/vmcai/AusserlechnerJK16}. 
Moreover, cutoffs also enable the \emph{synthesis} of implementations for parameterized systems from formal specifications~\cite{JacobsB14}, a problem closely related to learning.

In this paper, we develop a learning approach for BPs.
Given the expressiveness of BPs and the complexity of the general problem, we make some assumptions to keep the problem manageable.
In particular, we assume (1) that the BP under consideration has no hidden states, i.e., every state has at least one broadcast sending action by which it can be recognized; and {(2)} that there \emph{exists} a cutoff, i.e., a number $c$ such that the language (of finite words over actions) derived by $c$ processes is the same as the language derived by any number greater than $c$. We call such BPs \emph{fine}, and 
note that many BPs studied in the literature are fine. 
 
Moreover,
the restriction to fine BPs does not overly simplify the problem, as even under this assumption we obtain negative results for some basic learning problems.
We note that not all BPs have a cutoff (whether or not they have hidden states), and that when a cutoff exists the derived language is regular.
The fact that the derived language is regular also holds in previous works on learning concurrent models 
(communicating automata ~\cite{BolligKKL10}, workflow Petri nets (Esparza et al.~\shortcite{EsparzaLS11}), and negotiation protocols (Muscholl et al.~\shortcite{MuschollW22}))
 merely since a finite number of essentially finite state machines is in consideration.\footnote{The language of a BP in general need not be regular~\cite{DBLP:journals/tcs/FinkelS01,GeeraertsRB07} and this is true also with the restriction to no hidden states.}
We emphasize that this does not reduce the problem to learning a regular language, since the aim is to obtain the concurrent representation, which we show to be much more succinct than a  DFA for the language.  
 Moreover, the problem we consider goes way beyond what has been considered in previous works in the sense that
our approach works if \emph{there exists} a cutoff, but it does not require that the cutoff is known a priori, which is equivalent to the assumption of a (given) fixed number of processes in existing approaches.

We focus mainly on passive learning paradigms~\cite{delahiguera_2010}. Specifically, we consider the following problems: 
\begin{inparaenum}
\item \emph{Inference} --- given a sample consistent with a BP, return a BP that is consistent with the sample, 
\item \emph{Consistency} --- whether there exists a BP with at most $k$ states that agrees with a given sample,
\item \emph{Polynomial data} ---  whether characteristic sets are  of polynomial size,
and 
\item \emph{Polynomial Predictability} ---
whether
a learner  can correctly classify an unknown word with high probability after asking polynomially many membership and draw queries.
\end{inparaenum}

We prove a few properties of fine BPs relevant to learning in \S\ref{sec:properties}.
In \S\ref{sec:alg-inference}, we provide an inference algorithm that, given a sample of words that are consistent with a fine BP, infers a correct BP. 
 In \S\ref{sec:teachbility}, we show that the inference algorithm produces a minimal equivalent BP when the sample subsumes a characteristic set, and that characteristic sets of exponential size are unavoidable.
 In \S\ref{sec:conss-np-hard}, we show that consistency is NP-hard for the class of fine BPs.
Finally, in \S\ref{sec:BP-predict} we show that fine BPs are not polynomially predictable. Full proofs are available in the appendix.

\section{Preliminaries}

\subsection{Broadcast Protocols}
In the following we define broadcast protocols as introduced by 
 Emerson and Namjoshi~\shortcite{DBLP:conf/lics/EmersonN98} and studied in the seminal paper by Esparza et al.~\shortcite{EsparzaFM99}.
Broadcast protocols are one of the {most powerful} computational models for which some parameterized verification problems are still decidable, and are strictly more powerful than other standard communication primitives such as pairwise rendezvous or disjunctive guards~\cite{DBLP:conf/lics/EmersonK03}.

\paragraph{Broadcast Protocols (BPs)}
A \emph{broadcast protocol} $B=(S,s_0,L,R)$ consists of a {finite} set of states $S$ with an initial state $s_0 \in S$, a set of labels $L$
and a transition relation $R\subseteq S \times L \times S$, where $L = \{a!!, a?? \mid a \in A\}$ for some set of actions $A$. {A transition labeled with} $a!!$ is a broadcast \emph{sending  transition}, and a transition labeled with $a??$ is a broadcast \emph{receiving transition}, also called a \emph{response}.\footnote{Some models of BPs also consider rendezvous transitions, usually labeled with $a!$ and $a?$, but these can be simulated by broadcast transitions with a quadratic blowup in the number of states.} 
For each action $a\in A$, there must be exactly one outgoing response from every state.

    Given a BP $B=(S,s_0,L,R)$ we consider systems $B^n$, composed of $n$ identical processes that execute $B$. 
    Let $[n]$ denote the set $\{0,1,\ldots,n\}$.
    A \emph{configuration} of $B^n$ is a function $\textbf{q}: S \rightarrow [n]$, assigning to each state a number of processes.
    The \emph{initial configuration} $\textbf{q}_0$ is the configuration with $\textbf{q}_0(s_0)=n$ and $\textbf{q}_0(s)=0$ for all $s \neq s_0$.
    In a \emph{global transition}, all processes make a move: 
    One process takes a sending transition (labeled $a!!$), modeling that it broadcasts the value $a$ to all the others processes in the system.
    Simultaneously, all of the other processes take the receiving transition (labeled $a??$) from their current state.\footnote{We give a formal semantics of $B^n$ below.}

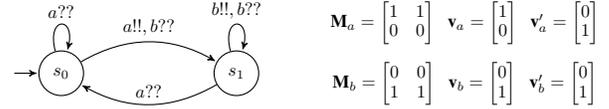
\begin{figure}
\centering
    \scalebox{0.675}{
    \begin{tikzpicture}[->,>=stealth',shorten >=1pt,auto,node distance=1.5cm,semithick,initial text=]
    
    \node[state,initial]  (s0)              {$s_0$};
    \node[state]         (s1) [ right of=s0, right= 1.5cm] {$s_1$};
    \node[label] (inner) [right  of=s1, node distance=4.5cm] {$~$};
    \node[label] (m0) [above of=inner, node distance=1cm]{$\textbf{M}_a =  
                \begin{bmatrix}
                1 & 1 \\
                0 & 0
                \end{bmatrix}$
\ \ $\textbf{v}_a = \begin{bmatrix}
                1  \\
                0
                \end{bmatrix}$
\ \ $\textbf{v}_a' = \begin{bmatrix}
                0 \\
                1
                \end{bmatrix}$
                
    };					
    \node[label] (m1) [below of=inner, node distance=.2cm]{$\textbf{M}_b =
                \begin{bmatrix}
                0 & 0 \\
                1 & 1
                \end{bmatrix}$
\ \ $\textbf{v}_b = \begin{bmatrix}
                0  \\
                1
                \end{bmatrix}$
\ \ $\textbf{v}_b' = \begin{bmatrix}
                0 \\
                1
                \end{bmatrix}$
                
    };

    \path (s0) edge      [bend left]             node {$a!!,b??$}    (s1);
    \path (s1) edge        [bend left, above]           node {$a??$}   (s0);		
    \path (s1) edge      [loop above]             node {$b!!,b??$}    (s1);
    \path (s0) edge [loop above] node {$a??$} (s0);
    \end{tikzpicture}}
\caption{Left: a simple BP. Right: its algebraic representation.}\label{bp3}
\end{figure}
    \commentout{        
    \begin{example}[BP representation]
    Consider the BP Fig.\ref{bp3}, state $s_0$ enables the sending-transition of action $a$, when $a$ is broadcasted a single processes moves from $s_0$ to $s_1$ while the rest of the processes respond with respect to their current state. Processes in state $s_0$ will respond by moving to $s_0$ ($s_0\xrightarrow[]{\text{a??}}s_0$), while those in $s_1$ will respond by moving to $s_0$.
    \end{example}}
    \begin{example}[A simple BP]
    Fig.\ref{bp3} (left) depicts 
     a simple BP $B$. 
    In the initial configuration of the system 
 $B^{9}$ we have $9$ processes in~$s_0$. 
    If a process broadcasts $a$, it moves to $s_1$ via the transition label $a!!$.
    The other processes respond following the $a??$ transition from their current state, 
    hence they remain
     in $s_0$. 
    Now we are in a configuration $\textbf{q}'$ with one process in $s_1$ and $8$ in $s_0$. 
    If from $\textbf{q}'$ another process broadcasts $a$, it moves from $s_0$ to $s_1$.
    The processes in $s_0$ stay there (following $a??$ from $s_0$), and the process in $s_1$ moves back to $s_0$ (following $a??$ from $s_1$). 
    Thus, we return to
    $\textbf{q}'$. If from $\textbf{q}'$ the process in $s_1$ broadcasts $b$ then it remains in $s_1$ (following $b!!$), 
    while all other processes move to $s_1$ via $b??$
    i.e., all processes will be in $s_1$.
    \end{example}

Following Esparza et al.~\shortcite{EsparzaFM99},
 we make the standard assumption that for each action $a$, there is a unique state $s_a$ {with an outgoing sending transition on $a!!$}. 
A state $s$ in a broadcast protocol is said to be \emph{hidden} if it has no outgoing sending transition. In this paper we consider broadcast protocols with no hidden states.
Note that the additional assumption of no hidden states is modest, since many examples from the literature satisfy it (e.g., the MESI protocol in Esparza et al.~\shortcite{EsparzaFM99} or the last-in first-served protocol in Delzanno et al.~\shortcite{DelzannoEP99}), and every protocol that does not satisfy this restriction can easily be modified to satisfy it without changing its functionality.

\paragraph{Semantics of $B^n$}
We can represent transitions of a system $B^n$ algebraically.
Assuming some ordering $s_0, s_1, ... s_{|S|-1}$ on the set of states $S$, we can identify configurations of $B^n$ with vectors from $[n]^{|S|}$, also called \emph{state-vectors}.
We use $\textbf{q}[i]$ to denote the entry in position $i$ of a state-vector \textbf{q}. 
For example, let $\textbf{u}_j$ be the unit vector with $\textbf{u}_j[j]{=}1$ and $\textbf{u}_j[i]{=}0$ for all $i \neq j$. Then the configuration where all $n$ processes are in $s_0$ is the vector $n \cdot \textbf{u}_0$. 
If $\textbf{q}$ is a state-vector with $\textbf{q}[i]\geq 1$ we say that $i$ is \emph{lit} in $\textbf{q}$. 
If state $s_i$ has an outgoing sending transition on action $a$ we say that $a$ is \emph{enabled} in $s_i$; if $i$ is lit in $\textbf{q}$ 
 we also say that $a$ is \emph{enabled} in $\textbf{q}$.

With each action $a$ we can associate 
\begin{inlineenum}
    \item 
    two unit-vectors $\textbf{v}_a = \textbf{u}_i$ for the origin and $\textbf{v}_a'=\textbf{u}_j$ for the destination, following 
    its sending transition $(s_i,a!!,s_j)$,  
    and  
    \item a \emph{broadcast matrix} $\textbf{M}_a$, which is an $|S|\times|S|$ matrix with $\textbf{M}_a(m,k)=1$ if there is a response $(s_k,a??,s_m) \in R$, and $\textbf{M}_a(m,k)=0$ otherwise. Every column of such a matrix is a unit vector.
\end{inlineenum}

Then {the transitions $T$ of $B^n$ are defined as follows:}
there is a transition between configurations $\textbf{q}$ and $\textbf{q'}$ on action $a$ in $B^n${, 
denoted $(\textbf{q},a,\textbf{q}') \in T$,}
iff there exists $(s_i,a!!,s_j)\in R$ with $\textbf{q}[i]\geq 1$ and:~
$\textbf{q}' = \textbf{M}_a \cdot (\textbf{q} - \textbf{v}_a) + \textbf{v}_a'.$

To see this, note that the state-vector $\textbf{q} - \textbf{v}_a$ corresponds to the sending process leaving the state $s_i$. The state-vector $\textbf{M}_a \cdot (\textbf{q} - \textbf{v}_a)$ describes the situation after the other processes take the responses on $a$. Finally,
$\textbf{q}'$ is the resulting state-vector after the sending process arrives at its target location.

\begin{example}
            Consider
            again the BP in Fig.\ref{bp3}, depicted with the broadcast matrices for the two actions $\textbf{M}_a$ and $\textbf{M}_b$,
            and the associated origin and destination vectors
            $\textbf{v}_a$, $\textbf{v}'_a$, $\textbf{v}_b$, $\textbf{v}'_b$.
            In configuration $\textbf{q} = \bivector{2}{2}$, $a$ is enabled ($s_0$ is lit).
                Computing the effect of $a$, we first get $\textbf{q} - \textbf{v}_a = \bivector{1}{2}$, then  
            $\textbf{M}_a \cdot 
                \bivector{1}{2} =
                \bivector{3}{0}$,
            and finally, $\textbf{q}' = \bivector{3}{0}
                 + \textbf{v}'_a = \bivector{3}{1}
                $. 

\end{example}

An \emph{execution} of $B^n$ is a finite sequence $e=\!\textbf{q}_0, a_1,\allowbreak \textbf{q}_1, a_2,  \ldots\!,\allowbreak a_{m}, \textbf{q}_m$ such that
$(\textbf{q}_i,a_{i+1},\textbf{q}_{i+1}) \!\in\! T$
for every $i \in [m\!-\!1]$. We say that $e$ is \emph{based on} the sequence of actions $a_1, \ldots, a_{m}$ and that $B^n(a_1\ldots a_m)=\textbf{q}_m$.
We say that a word $w \in A^*$ is \emph{feasible} in $B^n$ if there is an execution of $B^n$ based on $w$.
The \emph{language} of $B^n$, denoted $L(B^n)$, is the set of all words that are feasible in $B^n$, 
and the language of $B$, denoted $L(B)$, is the union of $L(B^n)$ over all $n\in\mathbb{N}$.

Let $B_1$ and $B_2$ be two BPs. 
We say that $B_1$ and $B_2$ are \emph{equivalent} iff $L(B_1){=}L(B_2)$.
A BP $B$ is said to have a \emph{cutoff} $k\in\mathbb{N}$ if for any $k'> k$ it holds that $L(B^k)=L(B^{k'})$.
A BP with no hidden states is termed \emph{fine} if it has a cutoff.
We measure the size of a BP by its number of states.
Thus, a BP is termed \emph{minimal} if there is no equivalent BP with fewer states.
Note that unlike the case of DFAs, there is no unique minimal fine BP, as shown by the following example. 

\begin{example}\label{ex:two-non-iso-bps}
Fig.\ref{fig:two-non-iso-bps} shows two BPs $B_1$ and $B_2$.
Note that $L(B_1^1)=a^*$, since with a single process, $a$ is the only possible action from $\textbf{q}_0$, and we arrive in $\textbf{q}_0$ after executing it.
With $2$ processes, after executing $a$ we arrive in state-vector $\bivector{1}{1}$, and we can execute either $a$ or $b$, and each of them brings us to $\bivector{1}{1}$ again.
Therefore, $L(B_1^2)=a(a\cup b)^*$.
Moreover, adding more processes does not change the language, i.e., $L(B_1^n)=L(B_1^2)$ for all $n$.
Hence $L(B_1)=L(B_1^2)$ and the cutoff of $B_1$ is $2$.
Similarly, we can show that $L(B_2)=a(a\cup b)^*$ and the cutoff of $B_2$ is $2$.
Note that $B_1$ and $B_2$ are not isomorphic, but they are equivalent.

\end{example}
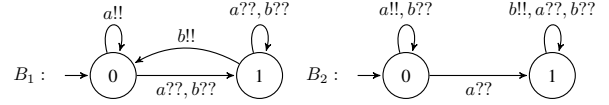
\begin{figure}
\centering
	\scalebox{0.695}{
		\begin{tikzpicture}[->,>=stealth',shorten >=1pt,auto,node distance=2.8cm,semithick,initial text=]
		
		\node[state,initial]  (s0)              {$\scstate{0}$};
		\node[state]         (s1) [right of=s0] {$\scstate{1}$};
            \node[label]          (ls) [left of=s0, node distance=1.6cm] {$B_1:$};
		\node[state,initial] (t0) [right of=s1] {$\scstate{0}$};
		\node[state]         (t1) [right of=t0] {$\scstate{1}$};		 \node[label]          (lt) [ left of=t0, node distance=1.6cm] {$B_2:$};				

		\path (s0) edge      [below]             node {$a??, b??$}    (s1);
		\path (s1) edge        [bend right, above]           node {$b!!$}   (s0);
		\path (s0) edge [loop above]  node {$a!!$} (s0);
		\path (s1) edge [loop above] node {$a??, b??$} (s1);					

		\path (t0) edge      [below]             node {$a??$}    (t1);
		\path (t0) edge [loop above]  node {$a!!, b??$} (t0);
		\path (t1) edge [loop above] node {$b!!, a??, b??$} (t1);					
		\end{tikzpicture}}
\caption{Two non-isomorphic fine BPs: $L\!(\!B_1\!)\!=\!L\!(\!B_2\!)\!=\!a(\!a\cup b\!)^*$.}\label{fig:two-non-iso-bps}
\end{figure}

{We note that the aforementioned examples from the literature (Esparza et al.~\shortcite{EsparzaFM99} and Delzanno et al.~\shortcite{DelzannoEP99}) also have a cutoff, and thus are fine BPs.}
\subsection{Learning problems}
A \emph{sample} for a BP $B$ is a set $\mathcal{S}$ of triples in $A^*\times\mathbb{N}\times \mathbb{B}$ where $\mathbb{B}=\{\true,\false\}$. A triple $(w,n,\true)$ (resp. $(w,n,\false)$) is consistent with $B$ if $w$ is feasible (resp. infeasible) in $B^n$. A sample is \emph{consistent} with $B$ if all triples in it are consistent with $B$. The size of $\mathcal{S}$ is defined as the sum of length of words in it.

We consider the following problems related to learning a class $\mathcal{C}$ of computational models, phrased for BPs.

\begin{problem}[Inference]
Devise an algorithm 
that given a sample $\mathcal{S}$ that is consistent with some BP in $\mathcal{C}$ returns a BP $B\in\mathcal{C}$ that is consistent with $\mathcal{S}$. We refer to such an algorithm as an \emph{inference algorithm}.
\end{problem}
Obviously one would prefer the returned BP to be minimal or sufficiently small. Phrased as a decision problem this is the consistency problem.

\begin{problem}[Consistency]
Given a sample $\mathcal{S}$ and  $k\in\mathbb{N}$  determine whether there exists a BP $B\in\mathcal{C}$ consistent with $\mathcal{S}$ with at most $k$ states. 
\end{problem}

The consistency problem is NP-hard even for DFAs~\cite{Gold78}. Thus, inference algorithms are often based on SAT or SMT solvers 
In \S\ref{sec:alg-inference} we provide such an inference algorithm for fine BPs.
Note that in many cases it is possible to devise a trivial inference algorithm (e.g., for DFAs the prefix-tree automaton) that is correct on the sample but makes no generalization and does not attempt to minimize the returned result.
We show in \S\ref{sec:teachbility} that if the sample is sufficiently complete (subsumes a \emph{characteristic set}), the inference algorithm we provide in fact returns the minimal BP that agrees with the sample. One may thus ask how large should a characteristic set be. We show that, unfortunately, it can be of size exponential in the number of states of the BP.

\begin{problem}[Polynomial data]
    Does there exist an inference algorithm $\algo{A}$ such that 
    one can associate with every BP $B\in\mathcal{C}$ a sample $\mathcal{S}_B$ of size polynomial in $B$ so that $\algo{A}$ correctly infers $L(B)$ from $\mathcal{S}_B$ or any sample subsuming it.
\end{problem}

\commentout{
The notion of \emph{learning in the limit using polynomial time and data}
a class $\mathcal{C}$ asks for the existence of a polynomial time inference algorithm $\algo{A}$ for $\mathcal{C}$  such that
for every $B\in\mathcal{C}$ it is possible to construct a sample $\mathcal{S}_B$ of size polynomial in the size of $B$ such that $\algo{A}$ returns a BP $B'$ for which $L(B')=L(B)$ when applied to any sample $\mathcal{S}$ that subsumes $\mathcal{S}_B$ and is consistent with $B$. {When this holds the sample $\mathcal{S}_B$ is termed the \emph{characteristic set} (CS) for $B$.}}
The last problem we consider is in the active learning paradigm. Its definition is quite long and deferred to \S\ref{sec:BP-predict}.

\begin{problem}
[Polynomial Predictability]
Can
a learner correctly classify an unknown word with high probability after asking polynomially many membership and draw queries.
\end{problem}
\section{Properties of Broadcast Protocols}
\label{sec:properties}
    Below we establish some properties regarding broadcast protocols that
    will be useful in devising the learning algorithm.  The proofs are given in Appendix~\ref{app:proofs-properties}.

It is not hard to see that the set of feasible words $L(B)$ of a given BP is prefix-closed. That is, if $uv$ is feasible for $u,v\in A^*$ then $u$ is feasible as well. Additionally, if $w\in A^*$ is feasible in $B^k$ then $w$ is feasible in $B^\ell$ for all $\ell > k$.  
    
    \begin{restatable}[Prefix-closedness and monotonicity]{lemma}{prefclose}\label{clm:monotonicity}
    If $B$ is a BP then $L(B)$ is prefix-closed. Moreover, $L(B^k)\subseteq L(B^\ell)$ for all $\ell > k$.
    \end{restatable}

The following lemma asserts that if $wa$ is feasible with $n$ processes but not with $m<n$, while $w$ is feasible with $m$ processes (for $w\in A^*$ and $a\in A$)
then $wa$ must be feasible with $m+1$ processes. 
Intuitively this is since $m$ processes are enough to execute all actions in $w$, and therefore any additional processes will take only receiving transitions along $w$, and will all arrive in the same local state. 
Thus, if $a$ is not enabled after $w$ with $m+1$ processes, then the additional process did not lit the state $s_a$ enabling $a$, and the same is true if we add any bigger number of processes.

\begin{restatable}[Step by step progress]{lemma}{stepstep}
\label{lem:progress-ce} 
Let $w\in A^*$, $a\in A$, and 
 $m<n$. 
If $w \in L(B^m)$ and $wa \notin L(B^m)$ yet $wa \in L(B^n)$, then $wa \in L(B^{m+1})$.
\end{restatable}

Recall that fine BPs have no canonical minimal representation, in the sense that, as shown in Ex.\ref{ex:two-non-iso-bps}, there could be two non-isomorphic BPs for the same language. The lack of a canonical minimal representation often makes it difficult to achieve a learning algorithm.
The following important lemma asserts, that while two minimal fine BPs may be non-isomorphic there is a tight correlation between them. 

Since every action is enabled by a unique state, and  every state enables at least one action,
in a minimal fine BP the set $A$ is partitioned between states, and if there is a state $s_1$ in  $B_1$ whose set of enabled actions is $A'=\{a_{i_1},a_{i_2},\ldots,a_{i_k}\}$ then there should be a state $s_2$ in $B_2$ for which the set of enabled actions is exactly $A'$.  
So we can define such a mapping between the states of two minimal fine BPs, and it must be that on every word $w$ if $\textbf{p}_w$ and $\textbf{q}_w$ are the state-vectors $B_1$ and $B_2$ reach after reading $w$, resp., then if state $s_1$ is lit in $\textbf{p}_w$ then the corresponding state $s_2$ (that agrees on the set of enabled actions) is lit in $\textbf{q}_w$. 
Moreover, the fact that $L(B_1)=L(B_2)$ guarantees that $L(B_1^m)=L(B_2^m)$ for any $m\in\mathbb{N}$. This bundle of claims can be proven together by induction first on the number of processes $m$, and second the length of the word $w$.  
In the following 
we use $\fact(s)=A'$ if $A'$ is the set of actions enabled in $s$.

\begin{restatable}[Relation between minimal fine equivalent BPs]{lemma}{twomin}\label{lem:iso}
Let $B_1$ and $B_2$ be minimal fine BPs with states $S_1$ and $S_2$ such that $L(B_1){=}L(B_2)$.
Then for every $m\in\mathbb{N}$ it holds that
$L(B_1^m){=}L(B_2^m)$ and
there exists a bijection $h:S_1\to S_2$ 
satisfying that 
$\fact(s){=}\fact(h(s))$ 
for any $s{\in} S_1$; and for any 
 $w\in A^*$
if $B_1^m(w)=\textbf{p}_w$ and $B_2^m(w)=\textbf{q}_w$
then
$\textbf{p}_{w}[i]$ is lit if and only if $\textbf{q}_{w}[h(i)]$ is lit, for every  state $i$. 
\end{restatable}

\section{Inferring a BP from a Sample}
\label{sec:alg-inference}
Let $\mathcal{S}$ be a sample. The inference algorithm $\algo{I}$ we devise constructs a BP $B_\mathcal{S}$ that agrees with $\mathcal{S}$. 

Let $A_\mathcal{S}$ be the set of actions that appear in $\mathcal{S}$ in at least one feasible word. 
In order to return a BP {$B_\mathcal{S}$} with no hidden states, we allow $B_\mathcal{S}$ to have a set of actions $A \supseteq A_\mathcal{S}$.%
\footnote{Note that in \S\ref{sec:CSgen} we show that it is enough to consider $A=A_\mathcal{S}$ if $\mathcal{S}$ is sufficiently complete.}
We use $S$ for the set of states of $B_\mathcal{S}$, and $s_0$ for its initial state.

We construct a set of constraints that define the BP $B_\mathcal{S}$. 
More precisely, we construct a set of constraints $\Psi_\mathcal{S}$ regarding the behavior of three partial functions $\fst:\! A \!\rightarrow \!S$, $\fsend: \!A \rightarrow\! S$, and  $\frec_a:\! S\! \rightarrow \!S$ for every $a \!\in \!A$
so that any valuation of these functions that satisfies $\Psi_\mathcal{S}$ implement a BP consistent with the sample.
{Formally,} we say that {functions $\fst, \fsend, \{\frec_a \mid a \in A\}$ \emph{implement}} a BP $B=(S,s_0,L,R)$
if for every $(s_i,a!!,s_j)\in R$ we have $\fst(a)=s_i$, $\fsend(a)=s_j$ and for every $(s_i,a??,s_j)\in R$ we have $\frec_a(s_i)=s_j$. We also use $\fact(s)=A'$ if $A'=\{ a\in A ~|~ \fst(a)=s\}$.

We turn to introduce some terminology regarding the sample.
Let $\mathcal{P}_{i}$ be the set of words $\{ w \colon\! (w,i,\true)\!\in \!\mathcal{S}\}$, and let  
$\mathcal{N}_{i}$ be $\{ w \colon\! (w,i,\false)\!\in \!\mathcal{S}\}$.
Note that by Lem.\ref{clm:monotonicity} it follows
that if $w \!\in\! \mathcal{P}_i$
then $w$ is feasible in $B^j$ for every $j\!\geq \!i$. Similarly, if
$w \!\in\!\mathcal{N}_i$,
then $w$ is infeasible in $B^j$
for every~$j\!\leq \!i$. We define $\mathcal{N}$ (resp. $\mathcal{P}$) as the union of all $\mathcal{N}_i$'s (resp. $\mathcal{P}_i$'s).

We define a relation between actions $a,b\in A_\mathcal{S}$ as follows. 
We say that $a\apart{\mathcal{S}} b$ if there exist a word $w\in A_\mathcal{S}^*$ and naturals $n'\geq n$ such that 
$(wa,n,\true)\in\mathcal{S}$ and $(wb,n',\false)\in\mathcal{S}$
or vice versa (switching the roles of $a$ and $b$).
Following Lem.\ref{lem:iso}, $a\apart{\mathcal{S}} b$ means that the sample $\mathcal{S}$ has information contradicting that $a$ and $b$ are enabled in the same state.
\begin{enumerate}
\item 
\label{cnstr:apart}
Our first constraints are therefore that for every $a,b\in A$ such that $a \apart{\mathcal{S}} b$ it holds that $\fst(a)\neq\fst(b)$.

\item 
\label{cnstr:no-hidden}
Since we identify states by the set of actions they enable and we assume there are no hidden states, we define the set of states $S = \{ \fst(a) \colon  a \in A \}$ as the set of terms $\fst(a)$.
This definition guarantees that no states are hidden. Moreover, we designate one of the states, that we term $s_0$, as the initial state: $\exists s\in S \colon s =s_0$.
\commentout{
Since we assume there are no hidden states, 
we need to demand that at least one action is associated with every state. 
We first collect the set of states as terms $\fst(a)$ that are the target of a sending or a receiving transition, and add to them the initial state, as follows:
$$\begin{array}{ll}
S\, =\! & \{ \fst(a) \colon \exists a, b\in A.\ \fsend(b)=\fst(a) \}\ \cup \\
& \{ \fst(a) \colon \exists a,b,c\in A.\ \frec_{c}(\fst(b))=\fst(a) \}\ \cup \\
& \{ \fst(a) \colon a \in \mathcal{P}_1 \}
\end{array}$$
Note that the last set refers to the initial state, since for a word of length one, $a$ must be fired from the initial state.
Now we demand that each state has an action:
$$\forall s \in S: \exists a \in A: \fst(a)=s.$$
}

\item \label{cnstr:initial}
The rest of the constraints are gathered by scanning the words first by length.
For every word of length one, i.e. action $a$, if $au \in \mathcal{P}_i$  for some $i$ and some $u\in A^*$ then we add $\fst(a)=s_0$, and if $a \in \mathcal{N}_i$ then we add $\fst(a)\neq s_0$.

\item \label{constr:longword-eq1}
Next, we scan inductively for every word $w \in \mathcal{P}$ for the minimal $i \geq 1$ such that $w \in \mathcal{P}_i$ and for every $w \in \mathcal{N}$ for the maximal $i \geq 1$ such that $w \in \mathcal{N}_i$.

In the base case $i{=}1$ we have
$w \in \mathcal{P}_{1} \cup \mathcal{N}_{1}$. Let $w=\allowbreak{a_1a_2{\ldots} a_m}$.
 We define $m{+}1$ variables
    $p_{0},p_{1},\ldots p_{m}$.
    The variable $p_{k}$ indicates the state the process reaches after the system reads $a_1\ldots a_{k}$, and $p_{0}$ should be $s_0$. 

If $w \in \mathcal{P}_1$, we add the constraint $\psi_{w,1}$ defined as 
$$\begin{array}{l}\left(p_0=s_0 \right) \wedge \bigwedge_{1\leq \ell \leq m} \left( p_{\ell-1} = \fst(a_{\ell}) \wedge  p_{\ell} = \fsend(a_{\ell})\right)
\end{array}$$

    This requires that the next letter $a_{\ell}$ is enabled in the state the process reached after $a_1a_2\ldots a_{\ell-1}$  was executed.
    
     If ${w\in\mathcal{N}_1}$  we add the following constraint
    $$\begin{array}{l}\displaystyle\bigvee_{0\leq \ell <m} \left(  \psi_{w[..\ell],1} \wedge p_\ell \neq \fst(a_{\ell+1}) 
    \right)
\end{array}$$

    where $w[..\ell]$ denotes the $\ell$'th prefix of $w$, namely $a_1a_2\ldots a_{\ell}$, and we let $\psi_{\epsilon,1}=\true$.
    This requires that at least one of the letters in the word is not enabled in the state the process reached, implying the entire word is infeasible with one process. 

    \item \label{cnstr:huge} 
    For the induction step $i>1$, let
     $w\in \mathcal{P}_{i}\cup \mathcal{N}_{i}$ and assume $w=a_1a_2\ldots a_m$. 
    We define $i(m+1)$ variables
    $p_{1,0},p_{2,0},\ldots p_{i,m}$.
    The variable $p_{j,k}$ indicates the state the $j$-th process reaches after the system reads $a_1\ldots a_{k}$.
    Accordingly, we set $p_{j,0}=s_0$ for every $1\leq j\leq i$.
    The {state} of the processes after reading the next letter, $a_{l+1}$, depends on their {state} after reading $a_{l}$.
    
    Let ${w{\in} \mathcal{P}_{i}}$, we add the constraint 
    $\psi_{w,i}$ defined as follows.
\[\psi_{w,i}=
\begin{array}
{l}\displaystyle\bigwedge_{1\leq \ell \leq m} \left( \bigvee_{1\leq j \leq i} \left(\left( p_{j,\ell-1} = \fst(a_{\ell})\right) \wedge \varphi_{j,\ell} \right)\right) 
\end{array}\]
where
\[\displaystyle\varphi_{j,\ell}=
\left( 
\begin{array}{c}
p_{j, \ell} = \fsend(a_{\ell}) \quad \wedge \\ 
    \displaystyle\bigwedge_{\scriptsize{\begin{array}{c}
    1\leq j' \leq i\\
    j' \neq j
    \end{array}}} p_{j',\ell} = \frec_{a_{\ell}}(p_{j',\ell-1})   
\end{array}
\right) 
\]
Intuitively, $\psi_{w,i}$ requires that for every letter $a_{\ell}$ of $w$ one of the processes, call it $j$, reached a state in which $a_{\ell}$ is enabled. The formula {$\varphi_{j,\ell}$} states that the $j$-th process took the sending transition on {$a_{\ell}$} and the rest of the processes took the respective receiving transition.  

    Let $w\in \mathcal{N}_{i}$. We then add the following requirement
    $$\begin{array}{l}    
    \displaystyle \bigvee_{0\leq \ell < m} \left( \psi_{w[..\ell],i} \wedge \bigwedge_{1\leq j \leq i} \left(p_{j,\ell} \neq \fst(a_{\ell+1})\right) \right)
    \end{array}$$
    where we let $\psi_{\epsilon,i}=\true$ for every $i$.
    
    Intuitively, if $w$ is infeasible with $i$ processes, then there exists a (possibly empty) prefix $w[..\ell]$ which is feasible with $i$ processes, therefore $\psi_{w[..\ell],i}$ holds, while $w[..\ell\!+\!1]$ is infeasible, meaning none of the $i$ processes is in a state where $a_{\ell+1}$ is enabled.

\end{enumerate}

\begin{theorem}\label{thm:inf-correct}
Let $\mathcal{S}$ be a sample that is consistent with some fine BP.
Let $B_\mathcal{S}$ be a BP that satisfies the prescribed constraints $\Psi_\mathcal{S}$. Then $B_\mathcal{S}$ is a BP consistent with $\mathcal{S}$.
\end{theorem}

\begin{proof}
We prove that if $w{\in}\mathcal{P}_i$ (resp. $w{\in}\mathcal{N}_i$) then $w$ is feasible (resp. infeasible) in $B_\mathcal{S}^i$, by induction first on the length of $w$ and then on $i$.
For $w$ of length $1$, this holds by the constraints in item~(\ref{cnstr:initial}).
Let~$w\!=\!a_1a_2\!\ldots\! a_n{\in}\mathcal{P}_i$. 
If $i{=}1$ then this holds by induction on $w$ thanks to constraint (\ref{constr:longword-eq1}).
Next we consider words of the form $w$ that are in $\mathcal{P}_{i}\cup \mathcal{N}_{i}$. 
If $w \in \mathcal{P}_i$ is already in $\mathcal{P}_{j}$ for $j < i$ then by the induction hypothesis it is already feasible for $j$ processes in the constructed BP, and by Lem.\ref{clm:monotonicity}, it is also feasible with $i$ processes. 
Otherwise, $w \in \mathcal{P}_{i}\setminus \bigcup_{j{<}i}\, \mathcal{P}_{j}$. In this case, constraint (\ref{cnstr:huge}) makes sure that every prefix of $w$ is feasible with $i$ processes, and requiring that for the next letter $a_\ell$
one of the $i$ processes reached the state enabling $a_\ell$ after reading the prefix up to $a_{\ell\!-\!1}$.

If  $w{\in}\mathcal{N}_{i}$ then 
$w$ is infeasible with $i$ processes.
 In this case, there exists a letter $a_\ell$ for $1\!\leq\! \ell \!\leq\! m$ such that while $w[..\ell{-}1]$ is feasible, $a_\ell$ is not enabled in any of the states that the $i$ processes reach after reading (the possibly empty) prefix $w[..\ell{-}1]$. This is exactly what constraint (\ref{cnstr:huge}) stipulates, which is added for $i$ or some $j > i$. 
 In the latter case, Lem.\ref{clm:monotonicity} implies that $w$ is infeasible in $B_\mathcal{S}^i$.
\end{proof}
Finally, note that our constraints are in the theory of equality with uninterpreted functions (EUF), and are therefore decidable.%
\footnote{While constraint (\ref{cnstr:no-hidden}) depends on the unknown set $A$, we can bound the size of $A$, e.g., by the size of the prefix tree of $\mathcal{S}$.}
Thus, an algorithm for inferring fine BPs can be implemented using an SMT solver.
\begin{corollary}
$\algo{I}$ is an inference algorithm for fine BPs.
\end{corollary}

\section{Returning a Minimal BP}
\label{sec:teachbility}
In this section we show that when the sample is sufficiently complete 
we can guarantee that we return a minimal equivalent BP, and not just a BP that agrees with the sample. We thus first show that every fine BP $B$ can be associated with a sample $\mathcal{S}_B$ so that there exists an inference algorithm $\algo{A}$ that when applied to any sample $\mathcal{S}$ that subsumes $\mathcal{S}_B$ and is consistent with $B$, returns a minimal fine BP that is equivalent to $B$. We refer to such a sample as a \emph{characteristic set} (CS). 

In the following, we first describe a procedure $\algo{G}$ that
generates a sample $\mathcal{S}_B$ from a fine BP $B$, and then we prove that an inference algorithm $\algo{A}$ can correctly infer a minimal BP $B'$ equivalent to $B$ from any sample subsuming $\mathcal{S}_B$.

\subsection{Generating a Characteristic Set}
\label{sec:CSgen}

The CS generation algorithm $\algo{G}$ 
builds a sequence of trees $\tree{T}_i$ starting with $i=0$ and incrementing $i$ by one until  $\tree{T}_{i+1}=\tree{T}_i$. The  edges of the tree are actions. The name of a node is taken to be the unique sequence of actions $w$ that leads to it. Thus, the root is named $\varepsilon$ and a child of a node $w\in A^*$ is named $wa$ for some $a\in A$.
A node $w\in A^*$ in tree $\tree{T}_i$ is annotated 
with $B^i(w)=\textbf{p}_{w,i}$, i.e. {the state-vector $B^i$ reaches when reading $w$}, if $w$ is feasible in $B^i$, and with the special symbol $\bot$ otherwise.  We call a node in the tree \emph{positive} if it is annotated with a state-vector, and \emph{negative} otherwise.
All nodes are either leaves or have exactly $|A|$ children. Negative nodes are always leaves.

The tree $\tree{T}_0$ consists of only a root $\varepsilon$ and is annotated with the state-vector of all zeros. 
The tree $\tree{T}_{i+1}$ is constructed from the tree $\tree{T}_i$ by first re-annotating all its nodes: The annotation of a positive $\textbf{p}_{w,i}$ is replaced by $\textbf{p}_{w,i+1}$, a negative node $w$ in $\tree{T}_i$ may become positive in $\tree{T}_{i+1}$ (if
$w$ is feasible with $i+1$ processes) and will be annotated accordingly with $\textbf{p}_{w,i+1}$.
Then we check, from every positive node, whether further exploration is needed.
A positive node will be declared a leaf if it is of the form  $va$ and it has an ancestor $u$, a prefix of $v$, for which $\textbf{p}_{u,i+1}=\textbf{p}_{v,i+1}$. Otherwise its $|A|$ children are constructed. 
That is, once we reach a node whose state-vector is the same as one of its ancestors, we develop its children, but the children are not developed further. 

The entire process terminates when $\tree{T}_{i+1}=\tree{T}_i$.\footnote{We say $\tree{T}_{i+1}=\tree{T}_i$ if they agree on the tree structure and the edge labels (regardless of the nodes' annotations).} 
Note that given that the BP has a cutoff, such an $i$ must exist. We use $\mathcal{T}$ for the last tree constructed, namely $\mathcal{T}_{i+1}$.
The sample is then produced as follows. For $n\in[i+1]$, let
$\mathcal{P}_n = \{(u,n,\true)~|~ n$ is the minimal for which $u$ is positive in $\tree{T}_n \}$, 
$\mathcal{N}_{n} = \{(u,n,\false)~|$ 
$n$ is the maximal for which $u$ is  negative  in $\tree{T}_n \}$.  Then the sample 
is 
$\mathcal{S}_B = \bigcup_{n=1}^{i+1} \left( \mathcal{P}_n \cup  \mathcal{N}_n\right)$. 

\subsection{Proving that $\algo{G}$ generates characteristic sets}
\label{sec:CSproof}

We first note that for any reachable state $s$ of the original BP, there exists at least one node $v$ in the tree where $s$ is lit (i.e. the entry for  $s$ in the state-vector annotating the node is lit). The following lemma strengthens this statement further.

\begin{restatable}{lemma}{reach}\label{lem:reachability-exhaust}
Let $\textbf{p}$ be a state-vector that is reachable in $B^n$.
Then for every shortest word $w$ that reaches $\textbf{p}$ in $B^n$
there exists a node $w$ in $\mathcal{T}_n$ such that $\textbf{p}_w=\textbf{p}$.
\end{restatable}
When the sample $\mathcal{S}$ subsumes 
the set $\mathcal{S}_B$ then $\apart{\mathcal{S}}$ 
induces an  equivalence relation between the actions: 

\begin{restatable}{lemma}{simeq}
\label{lem:sim-eq}
For two actions $a$ and $b$ define $a\sim_{\mathcal{S}} b$ iff it is not the case that $a\apart{\mathcal{S}} b$.
If $\mathcal{S}$ subsumes $\mathcal{S}_B$
then $\sim_{\mathcal{S}}$ is an equivalence relation.
\end{restatable}

\begin{theorem}\label{min_bp_sb}
Let $B$ be a fine minimal BP, and let $\mathcal{S}_B$ be the sample generated for it as above. There is an inference algorithm $\algo{A}$ such that if $B'$ is the result of $\algo{A}$  when applied to any set subsuming $\mathcal{S}_B$ and consistent with $B$ then $B'$ is minimal and $L(B')=L(B)$.
\end{theorem}
\begin{proof}
The inference algorithm $\algo{A}$ we use to prove this claim  runs in two steps.
First it runs a variation $\algo{I'}$ of the inference algorithm $\algo{I}$ presented in \S\ref{sec:alg-inference} 
that turns the constraint (\ref{cnstr:apart}) into an iff constraint. I.e. adding that $\fst(a)=\fst(b)$ unless $a\apart{\mathcal{S}} b$.
If running $\algo{I'}$ returns that there is no satisfying assignment then it runs $\algo{I}$. In both cases Thm.~\ref{thm:inf-correct} guarantees that the returned BP is consistent with the given sample. Therefore $\algo{A}$ is an inference algorithm.

Next we claim that if the given sample subsumes $\mathcal{S}_B$ then $B'$, the resulting BP, is minimal.
This holds since Lem.\ref{lem:sim-eq} ensures that 
$\apart{\mathcal{S}}$ defines the desired equivalence $\sim_\mathcal{S}$ between actions, and the revised 
constraint (\ref{cnstr:apart}) guarantees that actions are 
not enabled from the same state if and only if the sample separates them. (Note that any word consistent with the BP  cannot separate actions $a$ and $b$ if they are enabled from the same state.) Hence $\algo{I'}$ will not return that there is no  satisfying assignment.

Next we note that by Lem.\ref{lem:reachability-exhaust}
for every state-vector $\textbf{p}$ that is reachable in $B^m$ and for every shortest word $w$ that reaches $\textbf{p}$ in $B^m$
there exists a node $w$ in $\mathcal{T}_m$ such that $\textbf{p}_w=\textbf{p}$.
 If $w=a_1a_2\ldots a_n$ then for each $1\leq i\leq n$ one process took the sending transition $a_i!!$ and the rest of the processes responded with $a_i??$.  
Constraint (\ref{cnstr:huge}) makes sure the {assignment to}
 $\fst$, $\fsend$ and $\frec$ respect all the possible options that enabled this, making sure that for every two options for enabling $w$ that result in state-vectors $\textbf{p}_1$ and $\textbf{p}_2$, resp., the same states are lit in both $\textbf{p}_1$ and $\textbf{p}_2$. 
 
Hence, for any BP $B'$ that adheres to the constraints 
 there exists a mapping $h$ between the states of $B$ and $B'$
satisfying the requirements of Lem.\ref{lem:iso}. Thus, $L(B)\!=\!L(B')$. 
\end{proof}

Regarding the problem of \emph{polynomial data} we show that there exist fine BPs for which there is no CS of polynomial size. 
The proof constructs a family of BPs of size quadratic in $n$ for which there exists an action $a_\top$ such that the length of the shortest word containing $a_\top$ is exponential in $n$. 
Thus, any CS has to include at least one such long word.

\begin{restatable}{theorem}{csexp}\label{thm:cs-exp1}
There exists a family of fine BPs with no characteristic set of polynomial size.    
\end{restatable}
The same family used in the proof of Thm.\ref{thm:cs-exp1} also shows that fine BPs can be exponentially smaller than the minimal DFA accepting the same language. 
This is
since in a DFA for every state $q$ the length of the shortest word reaching $q$ is bounded by the number of states in the DFA.

\begin{corollary}
There exists a family of fine BPs for which the corresponding minimal DFA is of exponential size.
\end{corollary}

\section{Consistency is NP-Hard for fine BPs}
\label{sec:conss-np-hard}
We show below that  consistency is NP-hard even for fine BPs. We note that hardness is expected since DFA consistency is NP-hard~\cite{Gold78}, but it  does not directly follow from hardness of DFA consistency. This is since a DFA is not a special case of a fine BP.  However,  a fine BP can simulate a DFA, in a manner prescribed in Lem.\ref{lem:dfa-to-bp}.

Fig.\ref{fig:passive-np-hard} provides a schematic illustration of the simulation. The states in the rectangle are the original states of the DFA, and $\iota$ is the initial state of the DFA. All responses that are not shown in the figure are self-loops, we omit them to avoid clutter. 
In addition, to make the BP fine, we need to allow each state $q$ to enable some action, call it $q$. The simulation would like to ignore these actions, i.e. consider the projection of the words to words without these actions. Formally, let $\Gamma,\Gamma'$ be alphabets such that $\Gamma'\supseteq \Gamma$.
Let $w'$ be a word over $\Gamma'$.
We use $\pi_{\Gamma}(w')$ for the word obtained from $w'$ by removing letters in $\Gamma'\setminus\Gamma$.
If $B$ is a BP over {$A'$, such that} $A'\supseteq A$, we refer to the words in $\{\pi_A(w)~|~w\in L(B)\}$, abbreviated $\pi_A(L(B))$, as the $A$-feasible words of $B$. Lem.\ref{lem:dfa-to-bp} states in which manner the BP simulates the DFA.

\begin{figure}
\centering
	\scalebox{0.685}{
		\begin{tikzpicture}[->,>=stealth',shorten >=1pt,auto,node distance=1.8cm,semithick,initial text=]
  \tikzstyle{rec}=[rectangle,draw=black!50,thick,minimum width=6.25cm, minimum height = 1.0cm,node distance=4.35cm]
  \node[state, initial] (h1)  {$\scstate{i}$};
  
  \node[state] (i1) [right of = h1] {$\iota$};
  \node[rec] (d1) [right of = h1] {$D$};

  \node[state, accepting] (a1) [right of = i1, node distance = 5.15cm] {$q'$};
  \node[state] (a11) [left of = a1, node distance = 1.85cm] {$q''$};

  \node[state] (top) [ above of = a1, node distance = 1.38cm] {$\top$};
  \node[state] (bot) [ left of = top, node distance = 1.95cm] {$\bot$};
  \node[state] (a) [above of = h1, node distance = 1.38cm] {$\scstate{c}$};
  \node[state] (dol) [right of = a, node distance = 2.5cm] {$\scstate{x}$};

  \path (h1) edge      [above]  node {$i!!$}    (i1);
  \path (h1) edge      [right]  node {$i??$}    (a);
  
  \path (a1) edge      [left]  node {$\$??$}    (top);
  \path (a11) edge      [  left]  node {$\$??$}    (bot);
  \path (a) edge      [loop above]  node {$\Sigma!!$}    (a);
  \path (a) edge [above] node {$\$!!, \$??$} (dol);
  \path (dol) edge      [loop above]  node {$x!!$}    (dol);
  \path (bot) edge      [loop above]  node {$\bot!!$}    (bot);
  \path (top) edge      [loop above]  node {$\top!!$}    (top);
  \path (dol) edge      [above, bend left = 60]  node {$\top??$}    (top);
  \path (dol) edge      [above]  node {$\bot??$}    (bot);
  \path (i1) edge [left] node {$\$??$} (bot); 
\end{tikzpicture}}
\caption{Reduction from DFA-consistency to BP-consistency.}\label{fig:passive-np-hard}
\end{figure}
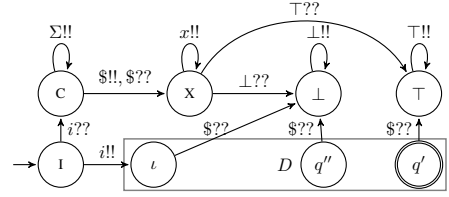

\begin{restatable}{lemma}{dfatobp}\label{lem:dfa-to-bp}
Let $L$ be a non-trivial regular language over $\Sigma$, and assume $n$ is the number of states in the minimal DFA for $L$.\footnote{A language is non-trivial if it is not the empty set or $\Sigma^*$.}
Let $A=\Sigma \cup \{ i, \$, \top, \bot, x\}$.
\begin{enumerate}[nosep]
    \item \label{lem:dfa-to-bp:feas-words}
    There exists a fine BP $B$ over actions $A'$ such that $A'\supseteq A$ satisfying 
    that 
    {
    $\pi_A(L(B^1))=\{i\}$ and 
    $\pi_A(L(B^n))\subseteq i (\Sigma)^* \$ x^* (\top^* \cup \bot^*) $ for every $n\geq 2$. 
    }
    \item \label{lem:dfa-to-bp:rel-to-dfa} 
    In addition, 
    for every $w\in \Sigma^*$:
    \begin{itemize}
        \item $w\in L \Leftrightarrow (iw\$\top, 2)$ is feasible in $B$.
        \item $w\notin L \Leftrightarrow (iw\$\bot, 2)$ is feasible in $B$.
    \end{itemize}
    \item \label{lem:dfa-to-bp:min-size}
    Moreover, for every $B$ satisfying the above it holds that $B$ has at least $n+5$ states.
\end{enumerate}
\end{restatable}
\begin{proof}[Proof sketch]

In order to keep the number of processes in the DFA part exactly one,
the $i$ action from the initial state of the BP (state $\scstate{i}$) sends one process to $\iota$ and the rest to state $\scstate{c}$. 
 To allow every letter $\sigma\in\Sigma$ to be taken from every state of the DFA,
the state $\scstate{c}$ of the BP enables all letters in $\Sigma$, and the original $(q,\sigma,q')$ transitions of the DFA are transformed into responses $(q,\sigma??,q')$. Next, to deal with the fact that the
language of a DFA is defined by the set of words that reach an accepting state whereas the language of a BP is the set feasible words, we introduce the letters $\$$, $\top$ and $\bot$. Upon $\$$ 
the single process in one of the DFA states moves into either state $\bot$ or state $\top$ depending on whether it was in an accepting state or a rejecting state; and all processes in states $\scstate{c}$ move to state $\scstate{x}$ where they wait to follow the respective $\bot$ or $\top$.  
 Now only $\{x,\top\}$ or $\{x,\bot\}$ are enabled. An $x!!$ transition would not change this situation whereas a $\top$ (resp. $\bot$) transition will ensure only $\top$ (resp. $\bot$) can be taken henceforth.
The complete proof, given in App.~\ref{proof:dfa-2-bp} shows the three requirements of the lemma are satisfied.
\end{proof}

\begin{restatable}{theorem}{bpnphard}\label{thm:bp-consistency-nphard}
BP consistency is NP-hard.
\end{restatable}
\begin{proof}[Proof Sketch]
The proof is by reduction from DFA consistency which was shown to be NP-hard in \cite{Gold78}.
Given an input to DFA-consistency, namely a sample $\mathcal{S}$ and a number $k$, we produce 
a sample $\mathcal{S}'$ and $k'=k+5$ as an input to BP-consistency so that the relation between the minimal BP consistent with
$\mathcal{S'}$ and the minimal DFA consistent with $\mathcal{S}$ is as described in Lem.\ref{lem:dfa-to-bp}. 

The sample $\mathcal{S}'$ consists of various triples ensuring  a BP consistent with $\mathcal{S}'$ has the structure given in Fig.\ref{fig:passive-np-hard}. For instance, $(i,1,\true)$ and $(ii,2,\false)$ enforce that $i$ is enabled in the initial state but is not a self-loop.  
A pair $(w,\true)$ (resp. $(w,\false)$) in $\mathcal{S}$ is altered to triple  $(iw\$\top\top,2,\true)$ (resp. $(iw\$\bot\bot,2,\true)$ in $\mathcal{S}'$, making sure that  words accepted (resp. rejected) by the DFA create respective feasible words ending with $\top$'s (resp. $\bot$'s). Each such pair carries  some additional triples added to $\mathcal{S}'$  to continue enforcing the structure of 
Fig.\ref{fig:passive-np-hard}.
\end{proof}
An alternative proof, inspired by (Lingg et al.2024) via a direct reduction from all-eq-sat is available in App.~\ref{proof:dfa-2-bp}.
\nocite{Lingg2024}

Note that given a BP $B$ and a pair $(w,n)\in A^*\times\mathbb{N}$ it is possible to check in polynomial time whether $w$ is feasible in $B^n$ by developing the state-vector $n \cdot \textbf{u}_0$ along the word $w$ in $B$. Consequently, and since a BP with $m$ states over set of actions $A$ can be described in size polynomial in $m$ and $|A|$, if $m$ is given in unary then BP-consistency is NP-complete.

\section{BPs are not polynomially predictable}\label{sec:BP-predict}
Here we show that fine BPs are not polynomially predictable with membership queries. The learning paradigm of polynomial predictability of a class $\mathcal{C}$ can be explained as follows. 
 The learner has access to an oracle answering \emph{membership queries} (\mq)  with regard to the target language $C\in\mathcal{C}$ or \emph{draw queries} (\dr) that can be implemented using \mq.
A membership query receives a word $w$ as input and answers whether $w$ is or is not in $C$. A draw query receives no inputs and returns a pair $(w,b)$ where $w$ is a word that is randomly chosen according to some probability distribution $D$ and $b$ is $\mq(w)$. We assume some bound $\ell$ on the length of the relevant examples, so that $D$ is a probability distribution on the set of relevant words. We assume the learner knows $\ell$ but $D$ is unknown to her.
At some point, the learner is expected to ask for a word whose membership it needs to predict, in which case it is handed a word $w$ (drawn randomly according to the same distribution $D$) and it should then answer whether $w$ is or is not in $C$.
We say that the class $\mathcal{C}$ is \emph{polynomially predictable} with membership queries, if given a bound $s$ on the size of the target language, the mentioned bound $\ell$ on the length of relevant examples, 
and an accuracy parameter $\varepsilon$ between $0$ and $1$,
there exists a  learner that will  classify the word to predict correctly with probability at least $(1-\varepsilon)$, after asking a number of queries that is polynomial in the size of the minimal BP of the target language.
We show that under plausible cryptography assumptions fine BPs (thus BPs in general) are not polynomially predictable.

\begin{figure}
\centering
	\scalebox{0.625}{
		\begin{tikzpicture}[->,>=stealth',shorten >=1pt,auto,node distance=1.55cm,semithick,initial text=]
  \tikzstyle{rec}=[rectangle,draw=black!50,thick,minimum width=4.675cm, minimum height = 1.05cm,node distance=3.53cm]
  \node[state, initial] (h1)   {$\scstate{h}_1$};
  \node[state] (g1) [right of = h1, node distance=1.65cm] {$\scstate{g}_1$};
  
  \node[state] (h2) [below  of = h1] {$\scstate{h}_2$};
  \node[state] (g2) [right of = h2, node distance=1.65cm] {$\scstate{g}_2$};
  \node[label] (h3) [below  of = h2, node distance=1.05cm] {$\circ\circ\circ$};
  \node[state] (hk) [below  of = h3, node distance=1.05cm] {$\scstate{h}_{k}$};
  \node[state] (gk) [right of = hk, node distance=1.65cm] {$\scstate{g}_k$};
  \node[state] (i1) [right of = g1, node distance=1.65cm] {$\iota_1$};
  \node[state] (i2) [right of = g2, node distance=1.65cm] {$\iota_2$};
  \node[state, accepting] (ik) [right of = gk, node distance=1.65cm] {$\iota_{k}$};
  \node[rec] (d1) [right of = g1] {$D_1$};
  \node[rec] (d2) [right of = g2] {$D_2$};
  \node[rec] (dk) [right of = gk] {$D_k$};
  \node[state] (a1) [right of = i1, node distance = 3.5cm] {$q'_1$};
  \node[state] (a2) [right of = i2, node distance = 3.7cm] {$q'_2$};
  \node[state, accepting] (a22) [left of = a2, node distance = 1.1cm] {$q''_2$};
 
  \node[state, accepting] (ak) [right of = ik, node distance = 3.7cm] {$q'_k$};
  \node[state ] (a2k) [left of = ak, node distance = 1.1cm] {$q''_k$};

  \node[state] (bot) [ right of = a2, node distance = 2.4cm] {$\bot$};
  \node[state] (s) [below of = hk, node distance = 1.4cm] {$\sstart$};
  \node[state] (c) [right of = s, node distance= 1.9cm] {$\scstate{c}$};
  \node[state] (dol) [right of = c, node distance = 7.6cm] {$\scstate{x}$};

  \path (h1) edge      [above]  node {$h_1!!$}    (g1);
  \path (g1) edge      [above]  node {$s??$}    (i1);
  \path (h1) edge      [left]  node {$h_1??$}    (h2);
  \path (h2) edge      [above]  node {$h_2!!$}    (g2);
  \path (g2) edge      [above]  node {$s??$}    (i2);
  \path (h2) edge      [left]  node {$h_2??$}    (h3);
  \path (h3) edge      [left]  node {$h_{k-1}??$}    (hk);
  \path (hk) edge      [above]  node {$h_{k}!!$}    (gk);
  \path (gk) edge      [above]  node {$s??$}    (ik);
  \path (hk) edge      [left]  node {$h_{k}??$}    (s);
  \path (a1) edge      [above]  node {$\$??$}    (bot);
  \path (a2) edge      [above]  node [yshift=-0.05cm]{$\$??$}    (bot);
  \path (a22) edge      [ right]  node {$\$??$}    (dol);
  \path (ak) edge      [above]  node {$\$??$}    (dol);
  \path (bot) edge      [right, bend left = 10]  node {$\bot!!\bot??$}    (dol);
  \path (s) edge      [above]  node {$s!!, s??$}    (c);
  \path (c) edge      [in=30,out=60, loop,looseness=7]  node [right, yshift=-0.1cm] {$\Sigma!!$}    (c);
  \path (c) edge [above] node [xshift = -0.5cm]{$\$!!, \$??$} (dol);
  \path (dol) edge      [loop right]  node {$x!!$}    (dol);
  \path (ik) edge [above] node [xshift=0.4cm, yshift =-0.55cm]{$\$??$} (dol);
  \path (i1) edge [above, bend left = 42] node {$\$??$} (bot);
  \path (i2) edge [above, bend left = 19] node [yshift = -0.06cm]{$\$??$} (bot);
  \path (a2k) edge [below, bend left = 20] node[xshift=0.2cm, yshift=0.15cm] {$\$??$} (bot);
		\end{tikzpicture}}
\caption{A BP simulating intersection of $k$ DFAs.}\label{fig:active-hard1}
\end{figure}

\begin{restatable}{theorem}{prednphard}\label{thm:bp-predic-hard}
Assuming the intractability of any of
the following three problems: testing quadratic residues
modulo a composite, inverting RSA encryption, or factoring Blum integers,
fine BPs are not polynomially predictable with $\mq$.
\end{restatable}
\begin{proof}[Proof Sketch]
The proof is via a reduction from the class $\mathcal{D}$ of intersection of DFAs, for which Angluin and Kharitonov have shown that $\mathcal{D}$ is not polynomially predictable under the same assumptions~\cite{AngluinK95}.
We show that given a predictor $\algo{B}$ for fine BPs we can construct a predictor $\algo{D}$ for the intersection of DFAs as follows.
Given a set 
 $D_1,D_2,\ldots,D_k$ of DFAs,
 we can construct a BP $B$ as shown in Fig.\ref{fig:active-hard1}
 such that $B$ simulates the run of the $k$ DFAs together.
 As in the proof of Lem.\ref{lem:dfa-to-bp} we can send one process to simulate any of the DFAs. Here we need $k$ processes to send a process to the initial state of each of the DFAs, and an additional process to enable all letters in $\Sigma$. Thus, the cutoff is $k+1$. The BP detects whether a given word $w$ is accepted by all the DFAs by checking whether $uw\$\bot$ is infeasible in $B$ where $u$ is some initialization sequence  that is required to send the processes to the initial states of the DFAs. 
\end{proof}

\section{Conclusion}
We investigated the learnability of the class of fine broadcast protocols.  
To the best of our knowledge, this is the first work on learning concurrent models  that does not assume a fixed number of processes interact.
 On the positive, we showed a passive learning algorithm that can infer a BP consistent with a given sample, and even return a minimal equivalent BP if the sample is sufficiently complete.
On the negative, we showed that the consistency problem for fine BPs is NP-hard; {characteristic sets may be inevitably of exponential size;} and the class is not polynomially predictable.

\bibliography{learn_bp.bib}
\appendix

\begin{center}
{\large{
    Proofs Appendix 
}}
\end{center}

\section{Complete Proofs for Section~\ref{sec:properties}}\label{app:proofs-properties}

\prefclose*

\begin{proof}[Proof of Lem.\ref{clm:monotonicity}]
    Prefix-closedness  holds since if $a_1a_2\ldots a_n$ is feasible, then for every $1\leq i\leq n$ the action $a_i$ is feasible after reading the prefix  $a_1a_2\ldots a_{i-1}$.
    
    For monotonicity, for two state vectors $\textbf{p}$ and $\textbf{q}$ we say that $\textbf{q}\geq \textbf{p}$ if for every $i\in[|S|]$ we have that $\textbf{q}[i]\geq \textbf{p}[i]$. Note that if $a\in A$ is enabled in $\textbf{p}$ then it is also enabled in $\textbf{q}$. 
    Let $w=a_1a_2\ldots a_m$ and let 
    $\textbf{p}_0, a_1, \textbf{p}_1, a_2, \ldots, a_{m}, \textbf{p}_m$ be the execution of $B^k$ on $w$.
    We can construct an execution $\textbf{q}_0, a_1, \textbf{q}_1, a_2, \ldots, a_{m}, \textbf{q}_m$ of $B^\ell$ by induction on the length of $w$ such that for every $i$ we have that $\textbf{q}_i>\textbf{p}_i$ entailing $w$ is feasible in $B^\ell$ as well.
\end{proof}

\stepstep*

\begin{proof}[Proof of Lem.\ref{lem:progress-ce}]
Suppose ${w \in L(B^m)}$ and ${{wa} \notin L(B^m)}$ yet ${wa} \notin L(B^{m+1})$. 
I.e., neither one of the processes that take some sending transitions when executing $w$ in $B^{m+1}$, 
nor one of the processes only taking receiving transitions (including the additional process compared to $B^m$) 
is in the state that has the sending transition on $a$. 
Since any further additional process will behave in the same way as the additional process in $B^{m+1}$, ${wa}$ can never become feasible. Contradicting that it is feasible in $B^n$.
\end{proof}
\twomin*
\begin{proof}[Proof of Lem.\ref{lem:iso}]
The proof is by induction on the number $m$ of interacting processes  and the length  of $w$.
For $m=1$ and $w=\epsilon$, we have
$\textbf{p}_\epsilon=\textbf{u}_{i}$ and $\textbf{q}_\epsilon=\textbf{u}_{j}$ for some $i,j{\in}\{0,\ldots,|S|\!-\!1\}$ where  $n{=}|B_1|=|B_2|$.
It follows that $s_i$ and $s_j$ are the initial states of $B_1$ and $B_2$, resp. Thus, $\fact(s_i)=\fact(s_j)$, as otherwise there is an action $a$ that is enabled in $B_1^1$ and not in $B_2^1$ or  vice versa (and therefore also in any $B_1^m$ and $B_2^m$ with $m {\geq} 1$), contradicting that $L(B_1)=L(B_2)$. Hence we can set $h(i)=j$ and the claim holds. 

For $m=1$ and $w=ua$ for $u\in A^*$ and $a\in A$,
by the induction hypothesis we know that for every state $i$, 
$\textbf{p}_{w}[i]$ is lit if and only if $\textbf{q}_{w}[h(i)]$ is lit.
Since we have only one process all global states are unit vectors. Thus, it must be that 
$\textbf{p}_{u}=\textbf{u}_{i}$ and $\textbf{q}_{u}=\textbf{u}_{j}$ for some  $i,j\in\{0,\ldots,|S|\!-\!1\}$.
That is,  $i$ and $j$ are the indices of the states in $B_1,B_2$ from which $a$ is enabled, resp.
Let $i',j'\in\{0,\ldots,|S|\!-\!1\}$ be the indices of the states that $B_1,B_2$ reach after reading $ua$. Then $\textbf{p}_{ua}=\textbf{u}_{i'}$ and $\textbf{q}_{ua}=\textbf{u}_{j'}$. Since this is true for any action $b$ that is feasible in $B_1$ or $B_2$ after $u$, $h(i')=j'$ satisfies the claim. Note that this also shows that any word $w$ that is feasible in $B_1^1$ is also feasible in $B_2^1$ and vice versa.

Assume the claim holds for $m$ we show it holds for $m+1$. Consider a word $w$. 
Let $\textbf{p}_w$ and $\textbf{q}_w$ be the state vectors $B_1^m$ and $B_2^m$ reach after reading $w$.
Then by the induction hypothesis, $L(B_1^m)=L(B_2^m)$ and for every state $i$ we have that
$\textbf{p}_{w}[i]$ is lit if and only if $\textbf{q}_{w}[h(i)]$ is lit.
Let 
$\textbf{p}'_w$ and $\textbf{q}'_w$ be the state vectors $B_1^{m+1}$ and $B_2^{m+1}$ reach after reading $w$.
Then $\textbf{p}_w$  and $\textbf{p}'_w$ are the same for every $i\in \{0,\ldots,|S|\!-\!1\}$ but one (and similarly for the $\textbf{q}$'s). For all of these indices the claim holds by the induction hypothesis. 
Let $j$ be the index with $\textbf{p}_w[j] \neq \textbf{p}'_w[j]$, i.e., $\textbf{p}'_w[j] = \textbf{p}_w[j]+1$.
If $\textbf{p}_w[j] \geq 1$, then by the induction hypothesis 
for every state $i$ we have that
$\textbf{p}_{w}[i]$ is lit if and only if $\textbf{q}_{w}[h(i)]$ is lit. In particular, this  holds for~$j$.

Otherwise, $\textbf{p}'_w[j]=1$, which implies that there is at least one action $a$ enabled from $\textbf{p}'_w$ that is not enabled from $\textbf{p}_w$.
By induction hypothesis, $L(B_1^m)=L(B_2^m)$, and therefore we know that $a$ is also not enabled from $\textbf{q}_w$.
Moreover, by Lemma~\ref{lem:progress-ce}, if $w \in L(B_2^m)$, $wa \in L(B_2)$ and $wa \notin L(B_2^m)$, then $wa \in L(B_2^{m+1})$. 
Thus, there must be a local state $k$ enabling $a$ such that $\textbf{q}_w[k]=0$ and $\textbf{q}'_w[k]=1$.
Hence, letting $h(j)=k$ satisfies the claim.
\end{proof}

\section{Complete Proofs for Section~\ref{sec:teachbility}} 

\reach*
\begin{proof}[Proof of Lem.\ref{lem:reachability-exhaust}]
The proof is by induction, first on $n$ then on the length of $w$.
For $n=1$, the construction of the tree clearly guarantees that all states reachable with one process have a respective node in the tree.

Suppose $\textbf{p}$ is reachable with $w$ in $B^n$ for $n>1$.
By Lemma \ref{lem:progress-ce}, there exists a prefix $u$ of $w$,
such that $u$ is feasible in $B^{n-1}$. By the induction hypothesis, $u$ is a node of $\mathcal{T}_{n-1}$. Suppose $w=ua_1a_2\ldots a_m$.
We can show by induction on $1\leq n$ that $ua_1a_2\ldots a_i$
is a node of $\mathcal{T}_n$ for every $i \leq m$ by simply following the tree construction.
\end{proof}
\simeq*
\begin{proof}[Proof of Lem.\ref{lem:sim-eq}]
Clearly $\sim_\mathcal{S}$ is symmetric and reflexive. To see that it is transitive, we first refer to Lem.~\ref{lem:reachability-exhaust} to deduce that there exists a node $v$ in $\mathcal{T}$ {in which $s$ is lit}. Let $n$ be the minimal for which $v$ is in $\tree{T}_n$. It follows that for every action $a$ that is feasible after $v$ with $n$ processes {we have} $(va,n,\true)\in\mathcal{S}$ and 
for every action $a$ that is infeasible after $v$ with $n$ processes we have $(va,n,\false)\in\mathcal{S}$.
Therefore, for any two actions we have $a\apart{\mathcal{S}}b$ iff $a$ and $b$ are not enabled from the same state, and $a\sim_\mathcal{S} b$ otherwise. Assume now that $a\sim_\mathcal{S} b$, $b\sim_\mathcal{S} c$. Then $a$ and $b$ are enabled from the same state, and $b$ and $c$ are enabled from the same state, implying $a$ and $c$ are enabled from the same state, i.e., $a\sim_\mathcal{S} c$ as required.
\end{proof}

\begin{figure}
\begin{center}
\scalebox{0.6}{
\begin{tikzpicture}[->,>=stealth',shorten >=1pt,auto,node distance=2.0cm,semithick,initial text=]
    \node[state,initial]          (i1)   {$\scstate{i}_1$};
    \node[state]    (i2) [below of=i1]  {$\scstate{i}_2$};	
    \node[state, color=blue,  fill=blue!10]   (out1)  [below of = i2]   {$\scstatep{1}$};
    \node[state, color=red, fill=red!10]   (help)  [left of=out1 , node distance = 3.85cm]   {$\scstate{h}_1$};
    \node[label, color=red]   (hh)  [below  of=help]   {$\begin{array}{c}\circ\\[-2mm] \circ\\[-2mm] \circ\end{array}$};
    \node[state, color=red, fill=red!10]   (helper)  [below of=hh]   {$\scstate{h}_{\ell}$};
    \node[state]   (sink)  [below of=helper]   {$\bot$};
		
    \node[state, color=blue,  fill=blue!10]   (out2)  [below  of = out1]   {$\scstatep{2}$};	
    \node[label, color=blue]   (outt)  [below  of=out2]   {$\begin{array}{c}\circ\\[-2mm] \circ\\[-2mm] \circ\end{array}$};
    \node[state, color=blue,  fill=blue!10]   (outer)  [below  of=outt]   {$\scstatep{n-1}$};
    \node[state, color=blue,  fill=blue!10]   (outer1)  [below  of=outer]   {$\scstatep{n}$};		
    \node[state, color=blue,  fill=blue!10]   (in1)  [right of=i2, node distance = 3.85cm]   {$\scstate{1}$};		
    \node[state, color=blue,  fill=blue!10]   (in2)  [below  of=in1]   {$\scstate{2}$};	
    \node[label, color=blue]   (inn)  [below  of=in2]   {$\begin{array}{c}\circ\\[-2mm] \circ\\[-2mm] \circ\end{array}$};
    \node[state, color=blue,  fill=blue!10]   (inner)  [below  of=inn]   {$\scstate{m}$};			
    \node[state]    (top) [below of=inner]  {$\top$};
    
    \path (i1) edge [bend left =25] node {$i_1!!$} (in1);
    \path (i1) edge [] node {$i_1??$} (i2);	
    \path (i2) edge [bend right = 20] node {$i_2!!$} (help);
    \path (i2) edge [] node {$i_2??$} (out1);
    \path (help) edge [left, color=red] node {$h_1!!$} (hh);
    \path (hh) edge [left, color=red] node {$h_{\ell -1}!!$} (helper);
    \path (helper) edge [left, color=red] node {$h_{\ell }!!$} (sink);
    \path (outer1) edge [bend right =35, left, blue!85] node {$H??$} (out1);
    \path (inner) edge [bend right =35, right, blue!85] node {$A??, H??$} (in1);
    \path (in1) edge [left, blue!85] node {$A??, H??$} (in2);
    \path (in2) edge [left, blue!85] node {$A??, H??$} (inn);
    \path (inn) edge [left, blue!85] node {$A??, H??$} (inner);
    \path (out1) edge [left, blue!85] node {$a_1??$} (out2);
    \path (out1) edge [left] node {$a_1!!$} (sink);
    \path (out2) edge [left, blue!85] node {$a_2??$} (outt);
    \path (out2) edge [above] node {$a_2!!$} (sink);
    \path (outt) edge [left, blue!85] node {$a_{n'-2}??$} (outer);
    \path (outer) edge [above] node {$a_{n'-1}!!$} (sink);
    \path (outer) edge [left, blue!85] node {$a_{n'-1}??$} (outer1);
    \path (outer1) edge [right] node {$c!!$} (sink);
    \path (inner) edge [left] node {$c??$} (top);
    \path (in1) edge [loop left, left, gray] node  {$b_1!!$} (in1);
    \path (in2) edge [loop left, left, gray] node  {$b_2!!$} (in2);
    \path (inner) edge [loop left, left, gray] node  {$b_m!!$} (inner);
    \path (top) edge [loop left, left] node  {$a_{\top}!!$} (top);
\end{tikzpicture}}
\end{center}
\caption{
The BP $B_{m,n,\ell}$ from a family
of fine BP's with quadratic cutoffs. }\label{fig:bp-complex-family}
\end{figure}
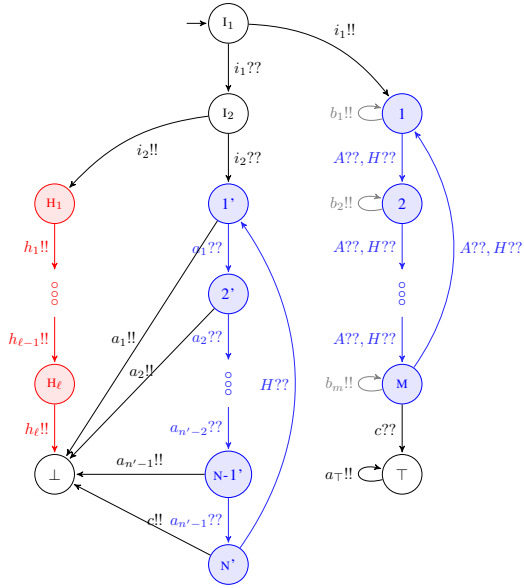

\begin{figure}[t]
\begin{center}
	\scalebox{0.6}{
		\begin{tikzpicture}[->,>=stealth',shorten >=1pt,auto,node distance=1.9cm,semithick,initial text=]
  \node[state, initial] (i0) {$i_0$};
  \node[state] (i1) [below of = i0] {$i_1$};
  \node[state] (i2) [below of = i1] {$i_2$};
  \node[state] (i3) [below of = i2] {$i_3$};
  \node[state, color=blue,  fill=blue!10] (111) [below of = i3] {$1''$};
  \node[label, color=blue]   (double)  [below  of=111]   {$\begin{array}{c}\circ\\[-2mm] \circ\\[-2mm] \circ\end{array}$};
  \node[state, color=blue,  fill=blue!10] (p3) [below of = double] {$5''$};
  \node[state, color=blue,  fill=blue!10] (11) [right of = i2, node distance = 1.95cm] {$1'$};
  \node[state, color=blue, fill=blue!10] (one) [below of = 11, node distance= 3.5cm] {$2'$};
  \node[state, color=blue, fill=blue!10] (p2) [below of = one, node distance= 3.5cm] {$3'$};
  \node[state, color=blue, fill=blue!10] (1) [right of = 11, node distance = 2.45cm] {$1$};
  
  \node[label,color=blue]   (zero)  [below  of=1, node distance= 3.5cm]   {$~$};
  \node[state, color=blue,  fill=blue!10] (p1) [below of = zero, node distance= 3.5cm] {$2$};
  \node [state] (d) [below of = p1] {$d$};
  \node[state]   (top)  [right  of=d]   {$\top$};
  \node[state, color=red, fill=red!10](h1) [left of = i3, node distance = 1.85cm] {$h_1'$};
  \node[state, color=red, fill=red!10](hl) [below of = h1] {$h_{2}'$};
  \node[state, color=red, fill=red!10](h11) [left of = i1, node distance = 3.5cm] {$h_1$};
  \node[state, color=red, fill=red!10](hdouble) [below of = h11] {$h_2$};
  \node[state, color=red, fill=red!10](hltag) [below of = hdouble] {$h_{3}$};
  \node[state] (sink) [below of= hltag, node distance = 7cm] {$\bot$};

  \path (i0) edge      [right]             node {$i_0??$}    (i1);
  \path (i0) edge      [bend left = 30, right]   node {$i_0!!$}    (1);
  \path (i1) edge      [right]             node {$i_1??$}    (i2);
  \path (i1) edge      [ above]   node {$i_1!!$}    (h11);
  \path (h11) edge [red!85] node {$h_1!!$} (hdouble);
  \path (hdouble) edge [red!85] node {$h_{2}!!$} (hltag);
  \path (hltag) edge  [left, color=red] node [] {$h_{3}!!$} (sink);
   \path (hltag) edge [left, bend left = 35, red!85] node {$H'??$} (h11);
   \path (hl) edge [left, color=red] node {$h'_{2}!!$} (sink);
  \path (h1) edge [color=red] node {$h'_{1}!!$} (hl);
  \path (i2) edge      [right]             node {$i_2??$}    (i3);
  \path (i2) edge      [above]   node {$i_2!!$}    (11);
  \path (i3) edge      [right]             node {$i_3??$}    (111);
  \path (i3) edge      [above]   node {$i_3!!$}    (h1);
  \path (111) edge      [left, blue!85]             node [xshift=0.12cm] {$a_1??$}    (double);
  \path (111) edge      [left]   node {$a_1!!$}    (sink);
  \path (double) edge      [left, blue!85]             node [xshift=0.12cm] {$a_{4}??$}    (p3);
  \path (double) edge      [right]   node {$a_{4}!!$}    (sink);
  \path(p3) edge [bend right = 30, blue!85] node [xshift=-0.15cm, right] {$\begin{array}{l}\!H?\!?\! \\ \! H'?\!?\!\end{array}$} (111);
  \path(11) edge [ blue!85] node [xshift=-0.15cm] {$\begin{array}{l}\!H?\!?\! \\ \! H'?\!?\! \\ \! A?\!? \end{array}$} (one);
  \path(one) edge [blue!85] node [xshift=-0.15cm] {$\begin{array}{l}\!H?\!?\! \\ \! H'?\!?\! \\ \! A?\!? \end{array}$} (p2);
  \path(p2) edge [right,bend right =30 , blue!85] node [xshift=-0.15cm]{$\begin{array}{l}\!H?\!?\! \\ \! H'?\!?\! \\ \! A?\!? \end{array}$} (11);
  \path(1) edge [ blue!85] node [xshift=-0.15cm]{$\begin{array}{l}\!H?\!?\! \\ \! H'?\!?\! \\ \! A?\!? \end{array}$} (p1);
  \path(p1) edge [right,bend right =30, blue!85] node[xshift=-0.15cm] {$\begin{array}{l}\!H?\!?\! \\ \! H'?\!?\! \\ \! A?\!? \end{array}$} (1);
  \path(p3) edge [below] node {$c!!$} (sink);
  \path(p2) edge [right] node {$c??$} (d);
  \path(p1) edge [right] node {$c??$} (d);
  \path(d) edge [above] node {$d!!$} (sink);
  \path(d) edge [above] node{$d??$} (top);
  \path(top) edge [loop above] node{$a_{\top}!!$} (top);
		\end{tikzpicture}}
\end{center}	
\caption{
The BP $P_5$ used to show a family of fine BPs with an exponential cutoff.
}\label{fig:bp-complex-family-n-5}
\end{figure}
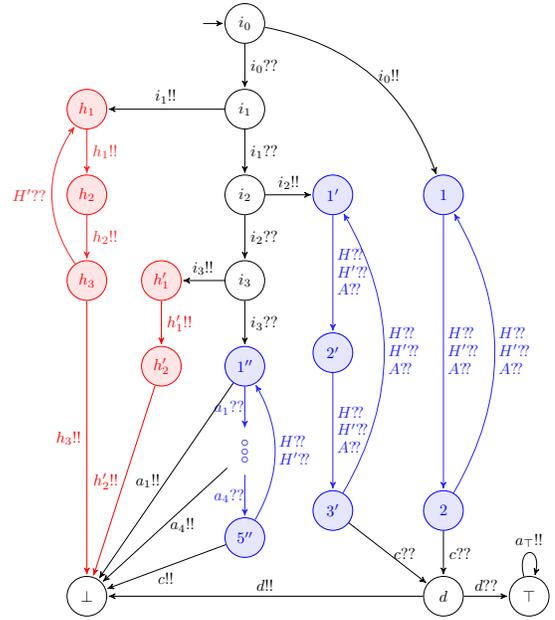

\commentout{
\begin{theorem}\label{thm:cs-exp}
There exists a family of fine BPs with no characteristic set of polynomial size.    
\end{theorem}}

Next we show that there exist fine BPs for which there is no characteristic set of polynomial size. For the sake of gradual presentation before proving Thm.\ref{thm:cs-exp1} we show that there exist fine BPs with cutoff of size quadratic in the size of the BP. 

\begin{proposition}
There exists a family of fine BPs which requires characteristic sets of quadratic size.
\end{proposition}
\begin{proof}
We adapt a family of BPs  used in~\cite{JaberJW0S20} for showing a quadratic cutoff for BPs without the restriction of no-hidden states, and for a slightly different definition of cutoff (reaching a particular state). The adaptation for no-hidden states is seamless. To work with our definition of cutoff we needed to introduce some auxiliary states. The family is given in Fig.\ref{fig:bp-complex-family}.

The family is parameterized by three natural numbers $m$, $n$, and $\ell$.
The BP $B_{m,n,\ell}$ 
has $n$ states in the left loop (colored blue), $m$ states in the right loop (colored blue as well), and $\ell$ helper states (colored red), and 4 additional states; overall $n+m+\ell+4$ states. 
The states $\scstate{i}_1$ and $\scstate{i}_2$ send one process to the right loop, one process to state $\scstate{h}_{1}$, and the rest of the processes to state $\scstatep{1}$.
The state $\top$ is used for identifying that at the same moment both states $\scstate{m}$ and $\scstatep{n}$ are lit together.

We use $H??$ and $A??$ as a shortcuts for $\{h_i??~|~i\in [1..\ell]\}$ and $\{a_i??~|~i\in [1..n'\!-\!1]\}$, respectively. 
From all of the states except for $\scstate{m}$ we assume $c??$ takes to the $\bot$ state, 
and $a_\top??$ takes all of the states to the $\top$ state (we didn't add these transitions to avoid clutter). 
Actions in 
$\{b_j \colon j\in [1..m]\}$ are added to satisfy the no hidden states assumption, where $b_j!!$ is a self-loop on state $\scstate{j}$ (colored gray), and  
$b_j??$ is a self loop for any state (not shown). One can see that traversing the left {blue} loop requires at least $n$ processes: one process for each of the $a_i$ transitions, for $1\leq i\leq n-1$, and one process for the $H$ transition. 
Each $H$ transition requires  one of the $h_i$ transitions to be taken. The structure of the $h_i$ sending transitions thus restrict the number of times $\scstatep{n}$ can be reached to $\ell+1$ (since it can be reached once without executing an $h_i$).  
In order to enable $a_\top$
there must be one process in state $\scstatep{n}$ that sends $c!!$ and one process in state $\scstate{m}$ responding to it. Therefore if $n$ and $m$ are co-prime, this can occur only once the left loop is traversed $m$ times (which requires $m$ helper states). Thus an overall of $n\cdot m$ processes are required to make a word using action $a_\top$ feasible. Accordingly, the shortest word using the $a_\top$ action is of length at least $n\cdot m$. This shows that there exist no characteristic sets for this family of size less than quadratic in the size of the given BP.
\end{proof}

To obtain an exponential size lower bound, instead of having two such loops, we have $k$ such loops
of sizes $p_1,p_2,\ldots,p_k$, where $p_i$ is the $i$-th prime number, as explained in the proof of Thm.\ref{thm:cs-exp1}.

\csexp*
 \begin{proof}[Proof of Thm.\ref{thm:cs-exp1}]
With every $n\in\mathbb{N}$ we associate a BP $P_n$ that has a loop for every prime smaller or equal to $n$. Denote these primes $p_1 \!<\! p_2\! <\!\ldots \!<\!p_k\! \leq\!  n$. Let $m\!=\!p_1\!\cdot \!p_2 \cdots p_k$ be the multiplication of all these primes. There is a special action $a_\top$ detecting synchronization among these loops, that is enabled from a state $\top$ that can be reached only once there is a process at the end of each of these $k$ loops.

Roughly speaking, the synchronization of the traversal around the different loops is controlled by the biggest loop (which sends $a_i!!$ actions), and the rest of the loops responding (following $a_i??$). 
To enforce a cutoff exists, there are additional helper states that limit the number of times biggest loop the ($p_k$) can be traversed to $m/p_k$ times, which is the minimum we need to reach $\top$. This is done by requiring an helper transition to traverse the transition closing the $p_k$ loop.

In order to avoid adding $m/p_k$ such helper states, we suffice with $p_1,p_2,\ldots,p_{k-1}$ helper states that also synchronize to count to $m/p_k$.
The BP for $n=5$ is given in Fig.~\ref{fig:bp-complex-family-n-5}. As before, $A??$, $H??$ and $H'??$ are abbreviations for all $a_i$, $h_i$ and $h'_i$ actions. 
You can see the loops corresponding to $p_1=2$, $p_2=3$, $p_3=5$ painted in blue.
The helper states are painted red.
There are additional auxiliary states $i_0,i_1,\ldots,i_{2(k-1)-1}$  (up to $i_3$ for $n=5$ for which $k=3$) enabling an initialization sequence 
$u_0=i_0i_1\cdots i_{2(k-1)-1}$  that is used to send one process to each of the {$k-1$ smaller} blue loops {i.e. the loops of $p_1,\ldots,p_{k-1}$} and one to each of the red loops/sequences. All additional processes {reach} the beginning of the biggest loop $p_k$ (for the example in Fig.\ref{fig:bp-complex-family-n-5}, it is state $\scstate{1}''$). 
{Auxiliary actions added to meet the no hidden states assumption are not shown.}

The sequence enabling $a_\top$ will take the following form.
Let $u=a_1a_2\cdots a_{p_k-1}$ and $v=u h_1  u h_2 u \cdots  u h_{p_{k-1}-1} u$. Then the desired sequence is  
$w=v h'_1 v h'_2 v\cdots h'_{p_1} v h_{p_2} u$.
After this sequence, all final states of the loops will be lit at the same time. Then, we can create a termination sequence $u_t$ ensuring there is a moment in which there is a process in the end of each of the $k$ loops. 
In the example of $P_5$ if the states {$\scstate{2}$, $\scstate{3}'$, $\scstate{5}''$} are all lit at the same time, then the sequence $cd$ will enable $a_\top$. In the general case if we have $k$ primes, the termination sequence would be of length $k-1$. 
Hence the sequence $u_0wu_ta_\top$ is feasible.

Note that the size of BP $P_n$ is quadratic in $n$.
Since all primes are of size $2$ at least, and since the number of primes of size $n$ or less is $\Theta(n/\log n)$ the number of 
processes required to enable $a_\top$ is {at least} $2^{\Theta(n/\log n)}$, and so is the size of a shortest word ($u_0wu_ta_\top$) that includes the $a_\top$ action.
Hence, a word of length exponential in the size of the BP is required to be in the sample, entailing this family has no characteristic set of polynomial size.    
\end{proof}

\section{Complete Proofs for Section~\ref{sec:conss-np-hard}}\label{proof:dfa-2-bp} 

\dfatobp*

\begin{proof}[Proof of Lem.\ref{lem:dfa-to-bp}]
We start with the first item.
Let $D=(\Sigma,Q,\iota,\delta,F)$ be a minimal DFA for $L$.
We build the BP $B=(A,S,s,R)$ as follows, see Fig.~\ref{fig:passive-np-hard}. 
The states are $S = Q \cup \{ \scstate{i} , \scstate{c},\scstate{x}, \top,\bot\}$, the initial state is $s\!=\!\scstate{i}$. The actions
$A\!=\Sigma\cup Q \cup\{i,\$,\bot,\top,x\}$.
The transitions are as follows:
$(\scstate{i}, i!!, {\iota})$, and $(\scstate{i}, i??, \scstate{c})$.
For every $\sigma\in \Sigma$:
$(\scstate{c}, \sigma!!, \scstate{c})$ and 
$(\scstate{c}, \$!!, \scstate{x}), (\scstate{c}, \$??, \scstate{x})$.
For every $(q,\sigma, q')\in\delta$ : $(q,\sigma??, {q'})$.
For every $q\in Q\setminus F$: $(q,\$??, \bot)$ and for every $q\in F$: $(q,\$??, \top)$.
Finally, $(\scstate{x},x!!,\scstate{x}),(\scstate{x},x??,\scstate{x})$,
$(\scstate{x},\top??,\top), (\scstate{x},\bot??,\bot)$,
$(\top,\top!!,\top)$ and $(\bot,\bot!!,\bot)$.
To adhere to the no-hidden states requirement we can add $(q, q!!,q)$ for all $q\in Q$ and every response that isn't defined is a self loop. 

First we note that with one process the set of $A$-feasible words is $i$.
With two processes, after executing $i$ we have one process in $\iota$ and one in $\scstate{c}$. The process in state $\scstate{c}$ can execute any word over $\Sigma$ and the other process would have to simulate the DFA on this word. After $\$$ is fired
a sequence of $x$'s is feasible followed by either a sequence of $\top$'s or a sequence of $\bot$'s, depending whether the DFA has reached an accepting state or not.
Note that with three  or more processes the same set of words is feasible, therefore the cutoff is $2$ and $B$ is fine.
It follows that $w\in L$ if and only if $(iw\$\top, 2)$ is feasible in $B$ and similarly 
 $w\notin L$ if and only if $(iw\$\bot, 2)$ is feasible in $B$. This proves the first two items.

For the third item, we first claim that the requirement of the $A$-feasible words require at least $5$ states.
Indeed, since with one process only $i$ is $A$-feasible, the only action feasible from the initial state, call it $\scstate{i}$ is $i$. Since $ii$ is infeasible, $i!!$ from the initial state must reach a new state, and since no other $A$-feasible words are feasible with one process, this state has no $A$-feasible actions. 
With two processes, after $i$ we can also fire all actions in $\Sigma \cup \{\$\}$, and so $i??$ from the initial state must reach a state that enables {all of them}, call it $\scstate{c}$.
Afterwards, once  a $\$$ has been fired neither $i$ nor $\Sigma$ or $\$$ are feasible, but only $x$ and either $\top$ or $\bot$. So $\$!!$ as well as $\$??$, must have reached new states that enables these.
Due to the fact that after $x$ nothing but $\top$ or $\bot$ is feasible, we get that $x$ is fired from a different state. Finally, since after $\top$ only $\top$'s are feasible, and similarly after $\bot$, these actions have to be enabled from different states as well, call them $\top$ and $\bot$. 
So there are at least $5$ states, and they possess transitions as shown in Fig.~\ref{fig:passive-np-hard}.

Now, assume towards contradiction that the minimal DFA for $L$ has $n$ states, and there is a BP $B$ satisfying these requirements with less than $n+5$ states. We claim that if $B$ has less than $n+5$ states then we can build a DFA for $L$ with less than $n$
states.
First note that as argued earlier, after firing $i$ there exists a process in a state, call it $\iota$, from which none of the actions in $A$ can be fired. 
Moreover, if there is more than one process, then the rest of processes after executing $i$ are in a different state, the one we called $\scstate{c}$.  
Assuming first, the responses on actions in $\Sigma$ do not reach one of the identified 5 states, it is clear that they must simulate the transitions of the DFA, as otherwise we would get a contradiction to the minimality of the DFA.
Removing this assumption, we note that out of the five states, responses on $\Sigma$ actions can only get to $\scstate{c}$ as otherwise the set of $A$-feasible words will not be a prefix of $i (\Sigma)^* \$ x^* (\top^* \cup \bot^*)$ as required. 
Assuming that this happens for some $\sigma\in\Sigma$, then at that point all the processes will be in $\scstate{c}$, which means that after taking $\$$ neither $\bot$ nor $\top$ is feasible, contradicting requirement (\ref{lem:dfa-to-bp:rel-to-dfa}). This concludes that the responses on actions in $\Sigma$ do behave the same as in the minimal DFA.
Regarding the actions in $A'\setminus A$, we can assume that sending and receiving transitions on every $a \in A'\setminus A$ behave such that, when starting from a configuration where only states $\scstate{c}$ and $\fst(a)$ are lit, then also after $a$ the same states are lit.\footnote{This means that either the transitions on $a$ are all self-loops, or the response from $\scstate{c}$ goes to $\fst(a)$, and at least one of the sending and receiving transitions on $a$ from $\fst(a)$ goes to $\scstate{c}$ (and is a self-loop otherwise).}
Otherwise we would get a contradiction either to requirement (\ref{lem:dfa-to-bp:rel-to-dfa}) or to minimality of the DFA:
\begin{itemize}
    \item First, note that if $\scstate{c}$ is not lit after $a$, then requirement (\ref{lem:dfa-to-bp:rel-to-dfa}) cannot be satisfied.
    \item If a transition on $a$ from $\fst(a)$ or $\scstate{c}$ goes to an auxiliary state other than $\scstate{c}$, then we also get a contradiction to requirement~(\ref{lem:dfa-to-bp:rel-to-dfa}).
    \item If a transition on $a$ from $\fst(a)$ goes to a non-auxiliary state other than $\fst(a)$, then we get a contradiction to minimality of the DFA, since then from this state exactly the same words in $\Sigma^*$ would be accepted as from $\fst(a)$. 
\end{itemize}
This concludes the proof.
\end{proof}

\bpnphard*

\begin{proof}[Proof of \ref{thm:bp-consistency-nphard}]
The proof is by reduction from DFA-consistency which was shown to be NP-hard in \cite{Gold78},
using the idea in Lem.\ref{lem:dfa-to-bp}. 
Given an input to DFA-consistency, namely a sample $\mathcal{S}$ and a number $k$, we produce 
a sample $\mathcal{S}'$ and $k'=k+5$ as an input to BP-consistency as follows.

We start by putting the triples $(i,1,\true)$, $(ii,2,\false)$, $(ix,2,\false)$, $(i\top,2,\false)$,
$(i\bot,2,\false)$. 
 Then for all ${\sigma \in \Sigma}$ we add $(i\sigma i,2, \false),$ $(i\sigma x,2, \false),$ $(i\sigma \$xi,2, \false)$ to $\mathcal{S}'$.
Then, for each $(w,\true)\in\mathcal{S}$ 
we add $(iw\$\top\top,2)$ and $(iw\$xx\top\top,2)$ with 
positive label to $\mathcal{S'}$, and the following words with 
negative label to $\mathcal{S'}$: $(iw\$\top\bot,2)$,
$(iw\$\bot\top,2)$, 
$(iw\top,2)$, $(iw\$\bot,2)$,
$(iw\$\top i,2)$,
$(iw\$\top x,2)$,
and $(iw\$\top\sigma,2)$, $(iw\$\sigma,2)$ for any $\sigma\in\Sigma$.

For each $(w,\false)\in\mathcal{S}$ 
we add $(iw\$\bot\bot,2,)$ and $(iw\$xx\bot\bot,2)$
with 
positive label to $\mathcal{S'}$,
and the following words with 
negative label to $\mathcal{S'}$: $(iw\$\top\bot,2)$, $(iw\$\bot\top,2)$, 
$(iw\bot,2)$, $(iw\$\top,2)$, $(iw\$\bot x,2)$,
and $(iw\$\bot\sigma,2)$, $(iw\$\sigma,2)$ for any $\sigma\in\Sigma$.

From similar arguments as the proof for the third item in  
Thm.~\ref{lem:dfa-to-bp} we can show that the reduction is valid,
namely there is a DFA for the given sample $\mathcal{S}$ with less than
$k$ sates if and only if there is a BP for the constructed sample $\mathcal{S}'$
with less than $k+5$ states.
\end{proof}

\subsection{Alternative proof that BP consistency is NP hard}
Below we provide a direct proof  by reduction from the problem of \emph{all-eq-3SAT} that is inspired by a recent proof on the hardness of DFA consistency~\cite{Lingg2024}. 

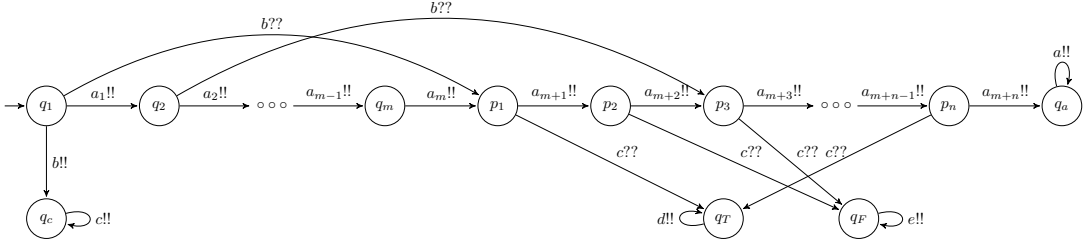
\begin{figure*}
\begin{center}
	\scalebox{0.6}{
		\begin{tikzpicture}[->,>=stealth',shorten >=1pt,auto,node distance=2.5cm,semithick,initial text=]
		
		\node[state,initial]          (q0)   {$\state{1}$};
		\node[state]    (qc) [below of=q0]  {$\state{c}$};		
		\node[state]   (q1)  [right of=q0 ]   {$\state{2}$};
		\node[label]   (q2)  [right  of=q1]   {$\circ\circ\circ$};
		\node[state]   (qmm1)  [right  of=q2]   {$\state{m}$};
		\node[state]   (qm)  [right  of=qmm1]   {$\statep{1}$};		
		\node[state]   (qm1)  [right  of=qm]   {$\statep{2}$};			\node[state]   (qm2)  [right  of=qm1]   {$\statep{3}$};			
		\node[label]   (qmc)  [right  of=qm2]   {$\circ\circ\circ$};		
		\node[state]   (qmn)  [right  of=qmc]   {$\statep{n}$};		
		\node[state]   (qmn1)  [right  of=qmn]   {$\state{a}$};		
		\node[state]   (qtrue)  [below of=qm2]   {$\state{T}$};												
		\node[state]   (qfalse)  [right  of=qtrue, node distance=3cm]   {$\state{F}$};

		\path (q0) edge                   node {$a_1!!$}    (q1);
		\path (q1) edge                   node {$a_2!!$}   (q2);
		\path (q2) edge  node {$a_{m-1}!!$} (qmm1);
		\path (qmm1) edge  node {$a_{m}!!$} (qm);
		\path (qm) edge  node {$a_{m+1}!!$} (qm1);
		\path (qm1) edge  node {$a_{m+2}!!$} (qm2);
		\path (qm2) edge  node {$a_{m+3}!!$} (qmc);  
		\path (qmc) edge  node {$a_{m+n-1}!!$} (qmn);
            \path (qmn) edge  node {$a_{m+n}!!$} (qmn1);  

		\path (q0) edge                   node {$b!!$}    (qc);				
		\path (qtrue) edge [loop left] node {$d!!$} (qtrue);
		\path (qfalse) edge [loop right] node {$e!!$} (qfalse);	
		\path (qc) edge [loop right] node {$c!!$} (qc);			
		\path (qmn1) edge [loop above] node {$a!!$} (qmn1);

		\path (q0) edge [bend left] 	node	{$b??$}(qm);
		\path (q1) edge [bend left] 	node	{$b??$}(qm2);		
		
		\path (qm) edge   	node	{$c??$}(qtrue);			
		\path (qm1) edge   	node	{$c??$}(qfalse);		
	      \path (qm2) edge   	node	{$c??$}(qfalse);		  
		\path (qmn) edge [above]  	node	{$c??$}(qtrue);								
		\end{tikzpicture}}
\end{center}	
\caption{A BP consistent with $\varphi=C_1\wedge C_2 \wedge \ldots \wedge C_{m}$ where 
$C_1=(x_1\vee x_2 \vee x_5)$, $C_2=(\overline{x_2}\vee\overline{x_3}\vee\overline{x_7})$ and the satisfying assignment
$x_1 = \true, x_2=\false, x_3=\false, \ldots,x_{n}=\true$. Responses that are not shown are self-loops.}\label{fig:bp-np-hard}
\end{figure*}

\bpnphard*

\begin{proof}
    The problem of 
    \emph{all-eq-3SAT} asks given an all-eq-3CNF formula $\varphi$ whether it has a satisfying assignment.  Where an all-eq-3CNF formula is a 3CNF formula where in each clause either all literals are positive or all literals are negative. This problem is known to be NP-complete.   
	 
    Let $\varphi=C_1 \wedge C_2 \wedge \ldots \wedge C_{m}$ be an all-eq-3CNF formula over a set of variables $V=\{x_1,x_2,\ldots, x_{n}\}$.
    We take the number $k$ to be $n+m+4$.
    For the alphabet of the sample we take $A=\{ a_i ~|~ 1\leq i \leq m+n\} \cup\{a,b,c,t,\!f\}$.
	We devise a sample $\mathcal{S}$ using four disjoint sets of words: $P_1$, $P_2$, $N_1$, $N_2$. The sample then consists of the triples $\{ (w,1,\true) ~|~ w\in P_1 \} \cup \{ (w,2,\true) ~|~ w\in P_2 \} \cup \{ (w,1,\false) ~|~ w\in N_1 \} \cup \{ (w,2,\false) ~|~ w\in N_2 \}$.
    The sets are defined as follows: 
    $$\begin{array}{l@{\,}l}
	P_1\,= & \{  a_1a_2\ldots a_{m+n}aa\} \cup \{bcc \}\\
	P_2\,= & \{a_1a_2\ldots a_ibctt ~|~ 1\leq i \leq m,\ \  C_i \text{ is positive} \}\ \cup \\  & \{a_1a_2\ldots a_ibc\!f\!f~|~ 1\leq i \leq m,\ \  C_i \text{ is negative} \}\\
    N_1\,= & \{ a_1a_2\ldots a_i a_j~|~ 1\leq i < m+n, \  j \neq {i+1} \}\  \cup\ \\ &  \{  a,ba,bb,c,t,\!f,bca,bct,bc\!f \}\   \\
    N_2\,= &  \{ a_1\ldots a_i b a_j a_{j+1} \ldots a_{m+n} ~|~ i \leq j \leq m \}\ \cup \\
		   & \{ a_1 \ldots a_i b a_{j+1} \ldots a_{m+n}~|~i \leq m,j > m,x_j\notin C_i \} \ \cup \\ 
     & \{a_1a_2\ldots a_ibct\!f ~|~ 1\leq i \leq m  \}\ \cup \\ 
     & \{a_1a_2\ldots a_ibc\!ft~|~ 1\leq i \leq m   \}
    \  \cup \\ 
    
     & \{a_1a_2\ldots a_i x~|~ 0< i \leq m+n,\ x \in \{c,t,\!f\} \}\ \cup \\
            & \{a_1a_2\ldots a_i a~|~ 0< i < m+n \}\ \cup \\
            & \{ a_1 \ldots a_i b a_{j+1} \ldots a_{k}a~|~i\leq m,\ j\geq m,\ k<m+n \}\ 
     \\
 \end{array}$$

    Suppose there is a satisfying assignment for $\varphi$. We can construct a BP with $k$ states consistent with the sample as shown in Fig.~\ref{fig:bp-np-hard}.
    The  states  $\state{1},\state{2},\ldots,\state{m}$ correspond to the clauses, the states $\statep{1},\statep{2},\ldots,\statep{n}$ correspond to the variables, the states $\state{T}$ and $\state{F}$ correspond to truth assignments $\true$ and $\false$, and there are two additional states $\state{c}$ and $\state{a}$.
    The initial state $\state{1}$ has transition with $b!!$ to $\state{c}$ which has a self transition on $c!!$.
    State $\state{a}$ has a self-loop on $a!!$.

    All states $\statep{j}$ corresponding to variables $x_j$ with a $\true$ (resp. $
    \false$) assignment have a $c??$ response to state $\state{T}$ (resp. $\state{F}$).
    For $1\leq i\leq m$, state $\state{i}$ corresponds to clause $C_i$ and to the word $a_1\ldots a_{i-1}$ (for $i=1$ the corresponding word is the empty word). 
    From each state $\state{i}$ we add a $b??$ transition to a state $\statep{j}$ corresponding to a satisfying literal in the clause of $C_i$.
    This makes the  words $a_1a_2\ldots a_ibctt$ (resp. $a_1a_2\ldots a_ibc\!f\!f$) feasible in $B^2$ for every positive (resp. negative) clause $C_i$ for which $x_j\in C_i$ and $x_j$'s assignment is $\true$ (resp. $\false$) conditioned a corresponding $b??$ transition from $C_i$ to $x_j$ exists. 
    It can be verified that this BP agrees with the sample.

Next, we show that if $\varphi$ has no satisfying assignment then any BP with $k$ states or less, disagrees with the sample. We show that the sample implies there are at least $k$ states, associated with actions as shown in Fig.\ref{fig:bp-np-hard}. Then, similar arguments regarding satisfaction of the clauses, imply that if no satisfying assignment exists, then any BP with these states would disagree with the sample. 

To see that such $k$ states are required, we note that for every BP consistent with $P_1$, if for some $\sigma \in A$ there is a word  in $P_1$ with $\sigma\sigma$ as an infix, there must be a state with self-loop on $\sigma$, since there is only one processes in the system. 
 Another thing to note is that if a word is feasible with $n$ processes then it is also feasible with any $m\geq n$ processes. In the other direction, if a word is infeasible with $n$ processes then it is also infeasible with any $m\leq n$ processes (see Lem.\ref{clm:monotonicity}). Moreover, if $w$ is feasible then any prefix of it is feasible.
 
The first set in $N_1$ prescribes that in any consistent BP, with a single process, the action $a_i$ may only be followed by $a_{i+1}$.
And due to $N_1$, with respect to $P_1$ this guarantees that any BP consistent with $\mathcal{S}$ has at least $m\!+\!n$ states, one for each of the $a_i$ actions.
Since $bcc$ is feasible yet $a_1a_2\ldots a_i c$ is infeasible for any $i \in [1,...,m\!+\!n]$ (see $N_2$, line 5), it follows that $b!!$ reaches a state where $c$ is enabled from and this state is different from any of the states $a_i$'s are feasible from. Hence, another state for the $c$ is needed.
Similar arguments taking that  $a_1a_2\ldots a_{m_n} a$ is in $P_1$ yet $a_1a_2\ldots a_i a$ is in $N_2$ for $i \in [1,...,m\!+\!n\!-\!1]$ (see line 6) show that another state for $a$ is required. Since $bcc$ is feasible (in $P_1$) while $bca$ is not (in $N_1$) the $a$ and $c$ action must be from different states. Hence so far we have identified $n\!+\!m\!+\!2$ states. Two more states are required for actions $t$ and $\!f$ since either $a_1a_2\ldots a_i bctt$  or  $a_1a_2\ldots a_i bc\!f\!f$  
in $P_2$ depending whether the clause $C_i$ is positive or negative, hence
$tt$ or $\!f\!f$ is always feasible after $a_1a_2\ldots a_ibc$, 
while  $a_1a_2\ldots a_i bct\!f$ and  $a_1a_2\ldots a_i bc\!ft$ are in $N_2$ thus $t$ and $\!f$ cannot be from the same state. Similar arguments show that these actions cannot be from the states the previous actions are. Thus $m\!+\!n\!+\!4$ states are required. With this structure, from similar arguments to the first direction, if the BP would agree with all the words in the sample, it would contradict that $\varphi$ is not satisfiable.
\end{proof}

\section{Complete Proofs for Section~\ref{sec:BP-predict}}\label{proof:bp-predict} 

\prednphard*
\begin{proof}[Proof of Thm.\ref{thm:bp-predic-hard}]
The proof is via a reduction from the class $\mathcal{D}$ of intersection of DFAs, for which Angluin and Kharitonov~\shortcite{AngluinK95} have shown that $\mathcal{D}$ is not polynomially predictable under the same assumptions.

We show that given a predictor $\algo{B}$ for fine BPs we can construct a predictor $\algo{D}$ for the intersection of DFAs as follows.
First, we show how to associate with any given set 
 $D_1,D_2,\ldots,D_k$ of DFAs  a particular BP $B$.
Let $D_i=(\Sigma,Q_i,\iota_i,\delta_i,F_i)$ and assume  without loss of generality that the states of the DFAs are disjoint and the DFAs are complete (i.e. from every state there is an outgoing transition on every letter {in $\Sigma$}).
We construct a fine BP $B=(A,S,s,R)$ with cutoff $k+1$
as follows. See illustration in Fig.\ref{fig:active-hard1}.

Let $Q=\cup_{i\in[1..k]}\,Q_i$.
The set of states $S$ is $Q  \cup \{\bot,\sstart, \scstate{c},\scstate{x}\} \cup \{\scstate{g}_i,\scstate{h}_i~|~i\in[1..k]\}$.
The initial state is $\scstate{h}_1$. 
Let $\Sigma_h=\{h_j~|~j\in[1..k]\}$ and $\Sigma_g=\{g_j~|~j\in[1..k]\}$ .
The set of actions  $A$ is $\Sigma \ \cup\ \Sigma_h\ \cup \Sigma_g \cup\ \{\$,\bot,x, s\}$.

The transitions are as follows:
for every ${i\in [1..k]}$ we have $(\scstate{h}_{i},h_i!!,\scstate{g}_i)$.
The receiving transitions on $h_i$ are $(\scstate{h}_i,h_i??,\scstate{h}_{i+1})$ for every $i\in[1..k\!-\!1]$,
and $(\scstate{h}_k,h_k??,\sstart)$.
For state $\sstart$, the transitions are as follows: $(\sstart, s!!, \scstate{c})$, $(\sstart, s??, \scstate{c})$.
For every  $\sigma\in\Sigma$ we have $(\scstate{c},\sigma!!, \scstate{c})$.
The transitions from the $\scstate{g}_i$'s states are $(\scstate{g}_i,s??,\iota_i)$ for every $i\in [1..k]$.
For every $(q,\sigma,q')\in\delta_i$ for $q,q'\in Q_i$ we have $(q,\sigma??,q')$.
For the $\bot$ state we have $({\bot},\bot!!, {\scstate{x}})$, $({\bot},\bot??, {\scstate{x}})$ and $(\scstate{x},x!!, {\scstate{x}})$.
Finally, for every $i\in[1..k]$ and 
 $q\in Q_i$ we have $(q,$\$??$,\scstate{x})$ if 
$q\in F_i$ and $(q,$\$??$,\bot)$ otherwise. 
Response transitions that are not specified are self-loops.

To satisfy the requirement of no-hidden states, we can assume that
the $\scstate{g}_i$ states have a self loop on $g_i!!$, and
every state $q$ of one of the DFAs has a self loop on $q!!$ (not shown in the figure to avoid clutter). The set of actions used would be $A\cup Q$.  
Note that this structure enforces the prefix to be of the form  $h_1u_1h_2u_2\ldots h_ku_k$
where $u_i \in \{g_1,g_2,\ldots,g_i\}^*$ and $k$ processes are required to enable $h_k$.
In order to enable $s$ an additional process is required. At this point any word $v$ from $\Sigma^*$ is feasible,
and there will be exactly one process in the initial state of each of the DFAs, simulating its run on $v$. Upon the letter $\$$ only $\bot$ or $x$ are feasible. Moreover, if there exists a DFA that rejects the word then $\bot$ is feasible, otherwise, no $\bot$ is feasible, only $x$.

Next we show how the predictor $\algo{D}$ for the intersection of DFAs uses the predictor 
 $\algo{B}$ for BPs to satisfy his task.
 Note that $\algo{D}$ has an oracle for the intersection of the DFAs $D_1,\ldots,D_k$ at his disposal, whereas  $\algo{B}$ expects answers regarding $B$. When $\algo{B}$ asks a \mq\ about $(w,n)$, i.e. whether $w$ is feasible in $B^n$ then $\algo{D}$ behaves as follows. 

Let $w'$ be the word obtained from $w$ by removing letters that are not in $A$.
If $w'$ is not a prefix of $h_1u_1h_2u_2\ldots u_{k-1}h_k u_k s \Sigma^* \$ \{\bot,x\}^*$  where $u_i\in \{g_1,g_2,\ldots,g_i\}^*$, or the number of $\bot$ letters exceeds one, then $\algo{D}$ returns the answer ``no" to $\algo{B}$.
Else, if $w$ is such a prefix and $w$ does not contain $\bot$, $\$$ or $x$, 
then $\algo{D}$ returns to $\algo{B}$ the answer ``yes" iff $n>k$ or $w'$ is a prefix of $h_1u_1h_2u_2\ldots u_{n-1}h_n u_n$ where the $u_i$'s are as above.
Otherwise, $w$ is such a prefix that contains $\$$.
If $n\leq k$ we return ``no". Otherwise,
 let $v$ be the maximal infix of $w'$ that is in $\Sigma^*$.
The predictor $\algo{D}$ asks a MQ about $v$. Assume it receives the answer $b$. 
If $w$ has $\bot$, it answers ``yes'' iff $b$ is ``no".
Otherwise it answers ``yes".

We claim that $\algo{B}$ receives correct answers. Indeed, the first couple of checks verify that $w$ is feasible according to the construction that builds $B$ from the given DFAs. For the last check, first note that if $w$ does not contain $\bot$ then $w$ is feasible. Otherwise, if $w$ does contain $\bot$ then it is infeasible if $v$ is accepted by all the DFAs
and is feasible if at least one DFA rejects it.

The next ingredient is to show how draw queries of $\algo{B}$ are simulated by $\algo{D}$. 
When $\algo{B}$ makes a $\dr$ query, then $\algo{D}$ makes a $\dr$ query. Suppose it receives the answer $(v,b)$. Then  
it passes $((h_1h_2\ldots h_k s v \$\bot,k+1),\neg b)$ to $\algo{B}$. 
Last, when $\algo{B}$ asks for a word $w$ to predict,
then $\algo{D}$ asks for a word $v$ to predict and flips the answer of the prediction of $\algo{B}$ on $(h_1h_2\ldots h_k s v \$\bot,k+1)$.
Note that if $v$ is accepted by the intersection of the DFAs, and $w=h_1h_2\ldots h_k s v \$\bot$ then $w$
is infeasible, and if $v$ is rejected then $w$ is feasible.
It follows that given $\algo{B}$ classifies the predicted word correctly then so does $\algo{D}$. Therefore, if $\algo{B}$ is a polynomial predictor for fine BPs then $\algo{D}$ is
a polynomial predictor for the intersection of DFAs.
\end{proof}

\end{document}